\documentclass[12pt]{article}
\usepackage{amsmath}
\usepackage{graphicx}
\usepackage{enumerate}
\usepackage{natbib}
\usepackage{hyperref}
\usepackage{amssymb}
\usepackage[section]{placeins}
\usepackage{algorithm}
\usepackage{algorithmic}
\usepackage{amsthm}
\usepackage{bbm}
\usepackage{titlesec}

\usepackage{url} 

\newcommand{\blind}{0}
\newcommand{\sign}{\text{sign}}

\newtheorem{theorem}{Theorem}

\newtheorem{lemma}[theorem]{Lemma}
\theoremstyle{definition}

\theoremstyle{remark}
\newtheorem{remark}{Remark}
\theoremstyle{assumption}
\newtheorem{assumption}{Assumption}
\titlelabel{\thetitle.\quad}

\makeatletter 
\@addtoreset{equation}{section}
\makeatother  

\begin{document}

\def\spacingset#1{\renewcommand{\baselinestretch}%
{#1}\small\normalsize} \spacingset{1}


\if0\blind
{
  \title{\bf Smoothed Concordance-Assisted Learning for Optimal Treatment
Decision in High Dimensional Data}
  \author{Angzhi Fan \thanks{
    The author would like to thank Professor 
Wenbin Lu and Professor Rui Song for their kind guidance during the author's visit to North Carolina State University in 2017.}\hspace{.2cm}\\
    Department of Statistics, University of Chicago}
  \maketitle
} \fi

\if1\blind
{
  \bigskip
  \bigskip
  \bigskip
  \begin{center}
    {\LARGE\bf Smoothed Concordance-Assisted Learning for Optimal Treatment
Decision in High Dimensional Data}
\end{center}
  \medskip
} \fi

\bigskip
\begin{abstract}
Optimal treatment regime is the individualized treatment decision rule which yields the optimal treatment outcomes in expectation. A simple case of treatment decision rule is the linear decision rule, which is characterized by its coefficients and its threshold. As patients’ heterogeneity data accumulates, it is of interest to estimate the optimal treatment regime with a linear decision rule in high-dimensional settings. Single timepoint optimal treatment regime can be estimated using Concordance-assisted learning (CAL), which is based on pairwise comparison. CAL is flexible and achieves good results in low dimensions. However, with an indicator function inside it, CAL is difficult to optimize in high dimensions. Recently, researchers proposed a smoothing approach using a family of cumulative distribution functions to replace indicator functions. In this paper, we introduce smoothed concordance-assisted learning (SMCAL), which applies the smoothing method to CAL using a family of sigmoid functions. We then prove the convergence rates of the estimated coefficients by analyzing the approximation and stochastic errors for the cases when the covariates are continuous. We also consider discrete covariates cases, and establish similar results. Simulation studies are conducted, demonstrating the advantage of our method.
\end{abstract}

\noindent%
\vfill

\newpage
\spacingset{1.5} 
\section{Introduction}
\label{sec:intro}
In order to give precised treatment decisions, we have to
take patients' heterogeneity into account. Decision rules based
on patients' own features are called \emph{individualized
treatment rules}. In a binary treatment
decisions setting, $\mathcal{A}=\{1,0\}$ are the treatment indicators and
$$\{(Y_{i},X_{i},A_{i}),i=1,2,...,n\}$$ are the \emph{i.i.d.}
observations, where $Y_{i}$ and $X_{i}$ are the outcome and features of subject $i$.
Suppose $Y_{i}^{*}(A_{i})$ is the expected outcome for subject $i$
with features $X_{i}$ and treatment $A_{i}$. If $Y_{i}^{*}(1)>Y_{i}^{*}(0)$, treatment $1$ is more favorable to subject $i$ compared to
treatment $0$. Assuming the linear decision rule $d(X)=\mathbbm{1}(\beta^{T}X>c)$, a treatment regime is determined by the coefficients $\beta$ and
the threshold $c$. Mean treatment outcomes can be modeled as
$$
E[Y^{*}(d(X))]=E_{X}[E(Y|X,A=1)d(X)+E(Y|X,A=0)(1-d(X))].$$

An \emph{optimal treatment regime} is the
treatment regime which yields the most favorable mean treatment outcomes. Optimal treatment decision with \emph{multiple decision timepoints} is called \emph{optimal dynamic treatment regimes}.
There are currently two leading dynamic treatment learning approaches. The first one is Q-learning by \cite{Watkins:1989}. \cite{qlearning1992} proved the convergence property of Q-learning. \cite{song2011} extended Q-learning into penalized Q-learning (PQ-learning). The second one is Advantage Learning (A-learning) by \cite{murphy2003}. \cite{Blatt2004ALearning} showed that A-learning is likely to have smaller bias than Q-learning.

Due to the rapid accumulation of patients' heterogeneity data, people started to
consider the high dimensional setting. \cite{HQlearning} extended Q-learning to high dimensional Q-learning, with a focus on the two-stage dynamic treatment regimes. \cite{PAL} proposed Penalized A-learning (PAL), which applied the Dantzig selector under the A-learning framework, with a focus on two or three-stage settings.  
 
For a \emph{single timepoint} treatment decision, \cite{ipwe}
used a doubly robust augmented inverse probability weighted estimator
to handle the possible model misspecification issue. \cite{zhao2012} adopted the support vector machine framework and proposed outcome weighted learning (OWL). \cite{powl}
modified OWL to be penalized outcome weighted learning (POWL). \cite{lu2013} introduced a robust estimator which doesn't require
estimating the baseline function and is easy to perform variable selection in large dimensions. CAL by \cite{CAL} used the contrast between two treatments to construct a \emph{concordance function}, and then used
the concordance function to estimate the coefficients $\beta^{*}$. CAL has a fast convergence rate, and it does not assume that the contrast function is a linear combination of the features. They assumed the stable unit treatment value assumption (SUTVA) in
\cite{rubin1980},  
\begin{align}
\label{equ:sutva}
Y_{i}^{*}(A_{i})=\mathbbm{1}(A_{i}=0)Y_{i}^{*}(0)+\mathbbm{1}(A_{i}=1)Y_{i}^{*}(1),
\end{align}
and the no-unmeasured-confounders condition,  \begin{align}
\label{equ:nuc}
\{Y_{i}^{*}(0),Y_{i}^{*}(1)\}\bot A_{i}|X_{i}.
\end{align}
The proposed concordance function $C(\beta)$ is
\begin{align}
\label{eqa:c_beta}
C(\beta) & =E\{[\{Y_{i}^{*}(1)-Y_{i}^{*}(0)\}-\{Y_{j}^{*}(1)-Y_{j}^{*}(0)\}]\mathbbm{1}(\beta^{T}X_{i}>\beta^{T}X_{j})\}\nonumber \\
 & =E\{[D(X_{i})-D(X_{j})]\mathbbm{1}(\beta^{T}X_{i}>\beta^{T}X_{j})\},
\end{align}
where the \emph{contrast function} $D(X_{i})=E[Y_{i}|A_{i}=1,X_{i}]-E[Y_{i}|A_{i}=0,X_{i}]$.
Within the concordance function, an important assumption is that $\beta^{*T}X_{i}$
is concord with $D(X_{i})$, which means
\begin{align}
\label{equ:concord}
D(X_{i})=Q(\beta^{*T}X_{i}),
\end{align}
where $Q(x)$ is an unknown monotone increasing function. Using the \emph{propensity score} $\pi(X_{i})=P(A_{i}=1|X_{i})$,
their $D(X_{i})$ is estimated by its unbiased estimator 
$$w_{i}=\frac{\{Y_{i}-v(X_{i},\hat{\theta})\}\{A_{i}-\pi(X_{i})\}}{\pi(X_{i})\{1-\pi(X_{i})\}},$$
where $v(X_{i,}\hat{\theta})$ can be any arbitrary function independent
of $A_{i}$, while it is usually chosen as the mean response of those
patients who received treatment $0$. The true $\beta^{*}$ and the true threshold $c_{0}$ in the linear decision rule are estimated by 
\begin{equation}
\hat{\beta}=\mathop{\arg\max}_{\beta:\|\beta\|_{2}=1}\frac{1}{n(n-1)}{\displaystyle \sum_{i\neq j}(w_{i}-w_{j})\times}\mathbbm{1}(\beta^{T}x_{i}>\beta^{T}x_{j}),
\end{equation}
and 
\begin{equation}
\hat{c}=\mathop{\arg\max}_{c}\frac{1}{n}{\displaystyle \sum_{i=1}^{n}\frac{Y_{i}\mathbbm{1}[A_{i}=\mathbbm{1}(\hat{\beta^{T}}x_{i}>c)]}{A_{i}\pi(x_{i})+(1-A_{i})[1-\pi(x_{i})]}}.
\end{equation}

In the high dimensional setting, however, due to the indicator function inside the concordance function,
the optimization of CAL is somewhat
difficult. In order to deal with this optimization issue of CAL, \cite{scal} proposed SCAL, which used hinge loss to replace $\mathbbm{1}(x)$, and they added
the $L_1$ penalty for variable selection purpose. They estimate $\beta^{*}$ and $c_{0}$ by 
$$
\hat{\beta}=\mathop{\arg\min}_{\beta:\|\beta\|_{2}=1}\frac{2}{n(n-1)}{\displaystyle \sum_{w_i>w_j}(w_{i}-w_{j})\times}(1-\beta^{T}(x_{i}-x_{j}))_{+}+\lambda{\displaystyle \|\beta\|_1},
$$
and 
$$
\hat{c}=\mathop{\arg\max}_{c}\frac{1}{n}{\displaystyle \sum_{i=1}^{n}w_{i}\mathbbm{1}(\hat{\beta^{T}}x_{i}>c)},
$$
It achieved $n^{1/2}$-rate but induced a relatively large bias due to the difference between the hinge loss and the original 0-1 loss.

Our way to overcome the optimization difficulties in CAL under high dimensional data is based on the smoothing method proposed by \cite{hanfang}. A family of sigmoid functions are used to substitute the
indicator function in CAL. The $L_1$ penalty is also added. We then employ the coordinate descent algorithm for optimization. Our method is called Smoothed Concordance-assisted Learning (SMCAL). 

Compared to SCAL, our SMCAL has a slower convergence rate but achieves
a much smaller bias. Numerical comparisons with POWL, SCAL and PAL demonstrated
the advantage of our method, especially when the sample size is relatively
large. Other than the continuous covariates cases in the SCAL settings, we
expand our SMCAL to the discrete covariates cases.

There are three main differences between our work and \cite{hanfang}. Firstly, they focused on the Maximum
Rank Correlation (MRC) proposed in \cite{MRC}, but our smoothing is applied on the concordance function.
Secondly, their paper studied Monotone Transformation model, but we focus on Monotone Single Index Model, and we have a faster convergence rate. Thirdly, we assume weaker assumptions instead of their normality assumption posted on the covariates, and
we expand our method to discrete cases.

The rest of this paper is organized as follows: Section \ref{sec:meth} introduces
SMCAL and describes the coordinate descent algorithm. Section \ref{sec:error_bound} displays
the $L_{2}$-error rate in continuous cases and the $L_{1}$-error rate in discrete cases, which are the main results of our paper. In Section \ref{sec:simulation}, we conduct numerical comparisons with SCAL, PAL and POWL. A real data application to STAR{*}D dataset is provided in Section \ref{sec:real_data}. The proofs of lemmas and theorems in Section \ref{sec:error_bound} are left to the supplementary material.

\section{Methods}
\label{sec:meth}

\subsection{The Model Set-up}

We still assume the Assumptions (\ref{equ:sutva}), (\ref{equ:nuc}) and (\ref{equ:concord}) posted in \cite{CAL}. Our smoothing procedure approximates their concordance function by
\begin{equation}
C_{n}(\beta)=E\{[D(X_{i})-D(X_{j})]F(\beta^{T}(X_{i}-X_{j})\alpha_{n})\},
\end{equation}
where we use a family of sigmoid functions
\[
F(\alpha_{n}x)=\frac{1}{1+e^{-\alpha_{n}x}}
\]
to replace $\mathbbm{1}(x)$ in \cite{CAL}. $\alpha_{n}$ is a positive constant depending on $n$. An unbiased estimator of $C_n(\beta)$ is
\begin{equation}
\hat{C_{n}}(\beta)=\frac{1}{n(n-1)}{\displaystyle \sum_{i\neq j}(w_{i}-w_{j})\times}F(\beta^{T}(x_{i}-x_{j})\alpha_{n}).
\end{equation}
Without loss of generality, we can set $\beta_{1}^{*}=\beta_{1}=1$. Define
\[
\beta_{r(\gamma),\alpha_{n}}^{*}=\mathop{\arg\min}_{\beta:\beta_1=1,\|\beta-\beta^{*}\|_{2}\leq r(\gamma)}(-C_{n}(\beta))
\]
where $r(\gamma)$ is a constant to be chosen. $\beta_{r(\gamma),\alpha_{n}}^{*}$ will converge
to $\beta^{*}$ as $\alpha_n$ goes to infinity. 

Define $\mathcal{H}=\{\beta\in \mathbb{R}^{d}:\beta_{1}=1,\|\beta\|_{0}\leq s_{n},\|\beta-\beta^{*}\|_{2}\leq r(\gamma)\}$,
where $s_{n}$ is a constant so that $\|\beta_{r(\gamma),\alpha_{n}}^{*}\|_0\leq s_n$. Notice that maximizing $\hat{C}_n(\beta)$ is equivalent to maximizing 
$$
\frac{2}{n(n-1)}{\displaystyle \sum_{w_i>w_j}(w_{i}-w_{j})\times}F(\beta^{T}(x_{i}-x_{j})\alpha_{n}),
$$
we can write our loss function as
\begin{equation}
l_{n}(\beta)=-\frac{2}{n(n-1)}{\displaystyle \sum_{w_i>w_j}(w_{i}-w_{j})\times}F(\beta^{T}(x_{i}-x_{j})\alpha_{n})+\lambda_{n}\|\beta\|_{1}.
\end{equation}
Then we use the following two steps to estimate $\beta^{*}$ and $c$. 
\begin{equation}
\hat{\beta}_{r(\gamma),\alpha_{n}}=\mathop{\arg\min}_{\beta\in \mathcal{H}}{\displaystyle l_{n}(\beta)}
\end{equation}
\begin{equation}
\hat{c}=\mathop{\arg\max}_{c}\frac{1}{n}{\displaystyle \sum_{i=1}^{n}w_{i}\mathbbm{1}(\hat{\beta}_{r(\gamma),\alpha_{n}}^{T}x_{i}>c)}
\end{equation}

Our $\hat{\beta}_{r(\gamma),\alpha_{n}}$ is different from $\beta_{r(\gamma),\alpha_{n}}^{*}$
due to the stochastic factors and the $L_1$ penalty. We will call $\beta_{r(\gamma),\alpha_{n}}^{*}-\beta^{*}$ the \emph{approximation error} and $\hat{\beta}_{r(\gamma),\alpha_{n}}-\beta_{r(\gamma),\alpha_{n}}^{*}$ the \emph{stochastic error}. 

\subsection{Coordinate Descent Algorithm}

We can iteratively use proximal gradient descent on each coordinate to solve the optimization problem. Step size $t$ is a fixed constant.
\FloatBarrier
\begin{algorithm}
\caption{Coordinate Descent Algorithm}
\begin{algorithmic}[1]
\STATE Initialize $\beta=(1,0,...,0)^{T}$ or $\beta=(-1,0,...,0)^{T}$. 

Repeat the following steps until converge:

For $k=2,3,\ldots,d$:

\STATE Calculate $g_k=\frac{-\alpha}{n(n-1)}{\displaystyle \sum_{i\neq. j}(w_{i}-w_{j})(x_{i,k}-x_{j,k})F^{'}((x_{i}-x_{j})^{T}\beta\alpha)}$.

\STATE Update $\beta_k=S_{t\lambda}(\beta_k-t g_k)$, where
$S_{t\lambda}(x)=$ $\begin{cases}
x-t\lambda & x>t\lambda\\
0 & -t\lambda\leq x \leq t\lambda\\
x+t\lambda & x<-t\lambda
\end{cases}$ is the soft-thresholding operator.
\end{algorithmic}
\label{alg:coor_des}
\end{algorithm}
Consider $-\frac{2}{n(n-1)}{\displaystyle \sum_{w_i>w_j}(w_{i}-w_{j})}F((x_{i}-x_{j})^{T}\beta \alpha_{n})$. When a pair $(i,j)$ satisfies $w_i>w_j$, likely we have $D(X_{i})>D(X_{j})$, which means $(x_{i}-x_{j})^{T}\beta^{*}>0$. Because $-F(x)$ is convex when $x>0$, we know that our optimization problem is likely to be convex when $\beta$ is close to $\beta^*$. $\lambda$ and $\alpha$ can be chosen by cross-validation.

\section{Convergence Rates of the Estimated Coefficients}
\label{sec:error_bound}

In this section, we establish the error rates for $\hat{\beta}_{r(\gamma),\alpha_{n}}-\beta^{*}$
under continuous and then discrete covariates cases. Let $\lambda_{n}\geq2||\nabla\hat{L}_{n}(\beta_{r,\alpha_{n}}^{*})||_{\infty}$. Assume the sparsity of the real coefficients
$\beta^{*}$. Let $S$ represent the nonzero indexes in $\beta^{*}$, $s=|S|$, and define $S^{c}=\{1,2,...d\}\setminus S$. Let $X_{i,I}$ represent all the entries of $X_{i}$ whose indexes belong to set $I$. The notation $\lesssim$ means that if $a_{n}\lesssim b_{n}$,
or equivalently $b_{n}\gtrsim a_{n}$, then there exist $C>0$ and $N_{0}>0$, such that for all $n>N_{0}$, we have $|a_{n}|\leq C |b_{n}|$. If $a_{n}\lesssim b_{n}$ and $b_{n}\lesssim a_{n}$, then we write $a_{n}\asymp b_{n}$.
\subsection{Continuous Cases}
\label{ssec:cts}
\subsubsection{Approximation Error}
\label{sssec:approximation}
To analyze the approximation error, we assume the following assumptions.
\begin{assumption}[]\label{asm:approx}$ $

(A1) $Q(x)$ is twice differentiable
and $|Q(x)|<M_{1}$.

(A2) For all $\beta:\:\beta_1=1,||\beta-\beta^{*}||_{2}\leq r(\gamma)$, and $t\in \mathbb{R}$, we have $|\frac{\partial}{\partial t}f_{X^{T}\beta}(t)|<M_{3}$.

(A3) Let $\Gamma_{(d-1)\times(d-1)}=-\nabla^{2}C(\beta^{*})$, where the derivative is taken with respect to every coordinate except the first one because $\beta_1^*=1$ is fixed. Assume $0<c_{3}\leq\lambda_{min}(\Gamma)\leq\lambda_{max}(\Gamma)\leq c_{4}$. 
\end{assumption}

\begin{remark}
(A2) is true for many distributions of $X$. For example, it's easy
to show that if $X\sim \mathcal{N}(\mu,\Sigma_{d\times d})$ and $0<c_{1}\leq\lambda_{min}(\Sigma)$, then (A2) is satisfied. 
\end{remark}

\begin{remark}
Since $\beta^*$ is the maximizer of $C(\beta)$, $\Gamma$ is non-negative definite. Assumption similar to (A3) also posed in Section 4.1.1,
\cite{hanfang}. But in the discrete cases, we will use another way to build the upper error bound and no longer need this assumption. 
\end{remark}

Under the above assumptions, the following Lemma \ref{lem:c_beta_approx_rate} measures the closeness between the concordance function $C(\beta)$ and its approximation $C_{n}(\beta)$.

\begin{lemma}
\label{lem:c_beta_approx_rate}
\begin{equation}
\sup_{\beta\in\mathbb{R}^{d}:\beta_1=1,\|\beta^{*}-\beta\|_{2}\leq r(\gamma)}|C(\beta)-C_{n}(\beta)|\lesssim\alpha_{n}^{-2}
\end{equation}
\end{lemma}

The convexity assumption (A3) in Assumption \ref{asm:approx} implies that
$\|\beta_{r(\gamma),\alpha_{n}}^{*}-\beta^{*}\|_{2}^{2}$ and $|C(\beta^{*})-C(\beta_{r(\gamma),\alpha_{n}}^{*})|$
are of the same rate, and thus the approximation error can be bounded through analyzing $|C(\beta^{*})-C(\beta_{r(\gamma),\alpha_{n}}^{*})|$, where we can apply Lemma \ref{lem:c_beta_approx_rate} and get Theorem \ref{thm:approx_rate}.

\begin{theorem}[Approximation Error Rate]\label{thm:approx_rate}
\begin{equation}
\|\beta_{r(\gamma),\alpha_{n}}^{*}-\beta^{*}\|_{2}^{2}\lesssim\alpha_{n}^{-2}
\end{equation}
\end{theorem}

\begin{remark}
\label{rmk:approx_rate}
The approximation error in \cite{hanfang} is 
\[
\|\beta_{r(\gamma),\alpha_{n}}^{*}-\beta^{*}\|_{2}^{2}\lesssim\alpha_{n}^{-1}\cdot\sup_{\beta:\beta_{1}=1}\frac{1}{\sqrt{2\beta^{T}\varSigma\beta}}.
\]
We have a faster convergence rate because we use the concordance function estimator 
\[
\frac{1}{n(n-1)}{\displaystyle \sum_{i\neq j}(w_{i}-w_{j})\times}\mathbbm{1}(\beta^{T}x_{i}>\beta^{T}x_{j})
\]
which is different from the Maximum Rank Correlation (MRC) estimator 
\[
\frac{1}{n(n-1)}{\displaystyle \sum_{i\neq j}\sign(w_{i}-w_{j})\times}\sign(\beta^{T}x_{i}-\beta^{T}x_{j})
\]
in their paper.
\end{remark}

\subsubsection{Stochastic Error}
\label{sssec:stochastic}
When investigating the stochastic error $\hat{\beta}_{r(\gamma),\alpha_{n}}-\beta_{r(\gamma),\alpha_{n}}^{*}$, we follow the framework of
Section 4.2 in \cite{hanfang}. In order to be consistent with their notations, define 
\[
\hat{L}_{n}(\beta)=-\hat{C}_{n}(\beta)=-\frac{1}{n(n-1)}{\displaystyle \sum_{i\neq j}(w_{i}-w_{j})\times}F(\beta^{T}(X_{i}-X_{j})\alpha_{n}),
\]
\[
L(\beta)=-C(\beta),\quad L_{n}(\beta)=-C_{n}(\beta),
\]
\[
\hat{\Delta}=\hat{\beta}_{r(\gamma),\alpha_{n}}-\beta_{r(\gamma),\alpha_{n}}^{*},
\]
and $P(\cdot)$ denotes $L_{1}$ Penalty, $P^{*}(\cdot)$ is the dual
norm of $P(\cdot)$. The following assumptions are needed in this subsection.

\begin{assumption}
\label{asm:stochastic}$ $

(A4)$|w_{i}|\leq M_{2}$. 

(A5)$\exists M_{4}$,
$s.t.$ $-M_{4}\leq\|X_{i}\|_{\infty}\leq M_{4}$. 

(A6)$\|\beta_{r(\gamma),\alpha_{n}}^{*}||_{0}\leq s_{n}$.
\end{assumption}

\begin{remark}
Our assumption (A4) is very similar to (A1) in Assumption \ref{asm:approx}, but (A1) doesn't include any stochastic factors.
\end{remark}

\begin{remark}
Although it's reasonable for us to post (A5) in our optimal treatment
decision problem, (A5) inevitably eliminates the multivariate normal
distribution. 
\end{remark}

\begin{remark}
(A6) appeared as Assumption (A0) in Section 4.2.2,
\cite{hanfang}. As we will see in the proof the Theorem \ref{thm:discrete_rate}, under some conditions, we have $\beta_{r(\gamma),\alpha_{n},S^c}^{*}=\Vec{0}$. So $s_{n}$ can be relatively small. 
\end{remark}

We provide two necessary steps, Lemma \ref{lem:stochastic_inf_norm} and \ref{lem:stochastic_error_rate}, in order
to apply Theorem 4.8 in \cite{hanfang}. Lemma \ref{lem:stochastic_inf_norm} is
used to bound the term $\langle\nabla\hat{L}_{n}(\beta_{r(\gamma),\alpha_{n}}^{*}),\Delta\rangle$.

\begin{lemma}
\label{lem:stochastic_inf_norm}
\begin{equation}
\|\nabla\hat{L}_{n}(\beta_{r(\gamma),\alpha_{n}}^{*})\|_{\infty}\stackrel{p}{\lesssim}\alpha_{n}\sqrt{\frac{log(d)}{n}}
\end{equation}
\end{lemma}

In \cite{hanfang}, they introduced a parameter $\kappa_{n}=E\mathop{\sup_{\beta\in\mathcal{H}}}|\hat{L}_{n}(\beta)-L_{n}(\beta)|$
and stated that $\kappa_{n}=O(n^{-1/2})$ when $d$ is fixed. However,
in our case $d$ can be much larger than $n$, so we will use a different
approach, which is our Lemma \ref{lem:stochastic_error_rate}, to analyze the rate of $\mathop{\sup_{\beta\in\mathcal{H}}}|\hat{L}_{n}(\beta)-L_{n}(\beta)|$. Define
$$B(h)=|\hat{L}_{n}(\beta_{r(\gamma),\alpha_{n}}^{*}-h)-L_{n}(\beta_{r(\gamma),\alpha_{n}}^{*}-h)-\hat{L}_{n}(\beta_{r(\gamma),\alpha_{n}}^{*})+L(\beta_{r(\gamma),\alpha_{n}}^{*})|,$$
the following Lemma \ref{lem:stochastic_error_rate} provides a bound for $B(h)$ in high probability.

\begin{lemma}
\label{lem:stochastic_error_rate}
\begin{equation}
{\displaystyle \sup_{\|h\|_{0}\leq2s_{n},\|h\|_{2}\leq2r(\gamma)}\frac{B(h)}{\|h\|_{2}}\stackrel{p}{\lesssim}\frac{\sqrt{s_{n}^{2}log(d)}}{\sqrt{n}}\alpha_{n}}
\end{equation}
\end{lemma}

$B(h)$ is used to bound the difference between $\hat{L}_{n}(\beta_{r(\gamma),\alpha_{n}}^{*}+\hat{\Delta})-\hat{L}_{n}(\beta_{r(\gamma),\alpha_{n}}^{*})$ and $L_{n}(\beta_{r(\gamma),\alpha_{n}}^{*}-\hat{\Delta})-L_{n}(\beta_{r(\gamma),\alpha_{n}}^{*})$. By using an argument similar to Theorem 4.8 and Lemma 4.10 in \cite{hanfang}, we can prove the stochastic error rate stated in the following theorem.

\begin{theorem}[Stochastic Error Rate]\label{thm:stochastic_rate}
\begin{equation}
\|\hat{\beta}_{r(\gamma),\alpha_{n}}-\beta_{r(\gamma),\alpha_{n}}^{*}\|_{2}^{2}\stackrel{p}{\lesssim}\alpha_{n}^{2}log(d)s_{n}^{2}/n+\alpha_{n}^{-2}
\end{equation}
\end{theorem}

From Theorem \ref{thm:approx_rate} and Theorem \ref{thm:stochastic_rate}, we conclude $\|\hat{\beta}_{r(\gamma),\alpha_{n}}-\beta^{*}\|_{2}^{2}\stackrel{p}{\lesssim}\alpha_{n}^{2}log(d)s_{n}^{2}/n+\alpha_{n}^{-2}$.
By choosing $\alpha_{n}\asymp(\frac{n}{log(d)s_{n}^{2}})^{\frac{1}{4}}$,
we get the overall $L_{2}$ error rate as

\begin{theorem}[Overall Error Rate]\label{thm:overall_rate}
\begin{equation}
\|\hat{\beta}_{r(\gamma),\alpha_{n}}-\beta^{*}\|_{2}\stackrel{p}{\lesssim}(log(d)s_{n}^{2}/n)^{\frac{1}{4}}
\end{equation}
\end{theorem}

\begin{remark}
For comparison, in \cite{hanfang}, they proved $\|\beta_{r(\gamma),\alpha_{n}}^{*}-\beta^{*}\|_{2}^{2}\lesssim\alpha_{n}^{-1}$
and $\|\hat{\beta}_{r(\gamma),\alpha_{n}}-\beta_{r(\gamma),\alpha_{n}}^{*}\|_{2}^{2}\stackrel{p}{\lesssim}\alpha_{n}^{2}log(d)s_{n}/n+\alpha_{n}^{-1}+\kappa_{n}$.
Choose $\alpha_{n}\asymp(\frac{n}{log(d)s_{n}})^{\frac{1}{3}}$, then in their paper $\|\hat{\beta}_{r(\gamma),\alpha_{n}}-\beta^{*}\|_{2} \stackrel{p}{\lesssim}(log(d)s_{n}/n)^{\frac{1}{6}}+\kappa_{n}^{\frac{1}{2}}$. 
\end{remark}

\subsection{Discrete Cases}
\label{ssec:discrete}

We will make some different assumptions as we discuss the discrete covariates cases in this section. The main difference between the discrete and continuous cases is that in the discrete cases, $-C(\beta)$ might be a constant around $\beta^*$.

\begin{assumption}
\label{asm:discrete}$ $

(B1) There exists a matrix $\mathcal{M}\in \mathbb{R}^{(d-s)\times (s-1)}$, such that the conditional distribution $(X_{i}-X_{j})_{S^{c}}-\mathcal{M}(X_{i}-X_{j})_{S\setminus 1} |(X_{i}-X_{j})_{S}$ is a symmetric distribution centered at the origin.

(B2) $|w_{i}|\leq M_{2}$.

(B3) $-M_{4}\leq\|X_{i}\|_{\infty}\leq M_{4}$.
\end{assumption}

\begin{remark}
In the continuous covariates cases, (B1) is satisfied under some multivariate normal distributions, where we can think of $S$ to be slightly larger than the nonzero indexes set. We can let $S$ include the features which are correlated with the first feature.
\end{remark}

\begin{remark}
(B2) and (B3) also appear in the continuous cases and are used in the proof of Lemma \ref{lem:stochastic_inf_norm} and \ref{lem:stochastic_error_rate}.
\end{remark}

In the discrete cases, we bound the approximation error and the stochastic error at the same time. Lemma \ref{lem:stochastic_inf_norm} and \ref{lem:stochastic_error_rate} also hold true for discrete cases and the proofs are almost the same. Moreover, we have a different main theorem as follows:

\begin{theorem}[Error Rate in Discrete Cases]\label{thm:discrete_rate}

As long as $s^{3}\log(d)/n\rightarrow0$, for all $(x_{i},x_{j})$ such that $(x_{i}^{T}-x_{j}^{T})\beta^{*}>0$, we have $(x_{i}^{T}-x_{j}^{T})_{S}\hat{\beta}_{r,\alpha_{n},S}\stackrel{p}{>}\frac{\epsilon}{\alpha_{n}}$, and $\|\hat{\beta}_{r,\alpha_{n},S^{c}}\|_{1}\stackrel{p}{\lesssim}(\frac{s^{3}\log(d)}{n})^{1/2}$.
\end{theorem}

\begin{remark}
In the proof, we choose $\alpha_{n}=max(1,\frac{log(c_{11})-log(\delta)+log(3)}{\varepsilon})$.
Using (B3), the theorem implies that for any $i,j$, $|\hat{\beta}_{r,\alpha_{n},S^{c}}(X_{i}^{T}-X_{j}^{T})_{S^{c}}|\stackrel{p}{\lesssim}(\frac{s^{3}\log(d)}{n})^{1/2}$. If $(x_{i}^{T}-x_{j}^{T})\beta^{*}>0$, then $(x_{i}^{T}-x_{j}^{T})_{S}\hat{\beta}_{r,\alpha_{n},S}\stackrel{p}{>}\frac{\varepsilon}{\alpha_{n}}\gtrsim(\frac{s^{3}\log(d)}{n})^{1/2}\stackrel{p}{\gtrsim}$$|\hat{\beta}_{r,\alpha_{n},S^{c}}(x_{i}^{T}-x_{j}^{T})_{S^{c}}|$. Therefore, when $\frac{s^{3}\log(d)}{n}\longrightarrow0$, we conclude
that once $(x_{i}^{T}-x_{j}^{T})\beta^{*}>0$, then $(x_{i}^{T}-x_{j}^{T})\hat{\beta}_{r,\alpha_{n}}\stackrel{p}{\geq}0$. In other words, if the contrast functions $D(x_i)>D(x_j)$, then $x_{i}^T\hat{\beta}_{r,\alpha_{n}}\stackrel{p}{\geq}x_{j}^T\hat{\beta}_{r,\alpha_{n}}$, which means we will rank the subjects in a right order with high probability.
\end{remark}

\section{Simulation Studies}
\label{sec:simulation}
Four methods are compared in this section: POWL in \cite{powl}, PAL in \cite{PAL}, SCAL in \cite{scal} and SMCAL. POWL method minimizes 
\begin{equation}
\frac{1}{n}{\displaystyle \sum_{i=1}^{n}\frac{Y_{i}\cdot[1-(2A_{i}-1)g(X_{i})]_{+}}{A_{i}\pi(X_{i})+(1-A_{i})[1-\pi(X_{i})]}}+\lambda{\displaystyle \sum_{i=1}^{d}}|\beta_{i}|
\end{equation}
where $g(X_{i})=\beta^{T}\cdot X_{i}+\beta_{0}$ is the decision function and the tuning parameter $\lambda$ is selected by the largest IPW estimator. PAL is a dynamic treatment regime approach but can be used in the single-stage problem because PAL estimates the latest stage at first. PAL perform variable selection using penalized A-learning, then use unpenalized A-learning to solve the coefficients. We used the R package provided by the authors of PAL to get its results.

To evaluate the estimated coefficients, we report the MSE after normalization. To evaluate the variable selection accuracy, we report the Incorrect Zeros (true coefficient is zero but the estimation is nonzero) and Correct Zeros (true coefficient is zero and the estimation is zero). To evaluate the estimated treatment regime, we report the Percentage of Correct Decision (PCD) and the mean response if we follow the estimated treatment
regime (Estimated Value). Estimated Value is calculated by drawing 1000 new samples. Standard errors of the PCDs and Estimated Values are in the parenthesis. 

\subsection{Low Dimension}
\label{ssec:low-dim}
\subsubsection{Linear Case}
\label{sssec:linear}
The first example can be found in \cite{zhao2012} as well as in \cite{scal}. Set $d=50$, $X_{i1},X_{i2},...,X_{i50}$ are
independent variables, all generated from $U(-1,1)$ in the continuous cases. $\pi(X_{i})$
is chosen to be $0.5$ for all $i=1,2,...n$. $Y_{i}|X_{i},A_{i}\sim N(Q_{0}(X_{i}),1)$,
where $Q_{0}(X_{i})=1+2X_{i1}+X_{i2}+0.5X_{i3}+0.442(1-X_{i1}-X_{i2})(2A_{i}-1)$. $100$ simulations are conducted for $n=30,100,200$ respectively. 

In the discrete cases, everything remains the same except that we generate $X_{i1},X_{i2},...,X_{i50}$ from the uniform distribution on $\{-0.9, -0.7, -0.5, -0.3, -0.1, 0.1, 0.3, 0.5, 0.7, 0.9\}$. 

Numerical comparisons are summarized in Table \ref{tab:low}, where we did simulations only for SMCAL and PAL, the results by SCAL and POWL were from \cite{scal}.
\begin{table}
\caption{Comparison between POWL, PAL, SCAL and SMCAL: low dimensional case.  \label{tab:low}}
\small
\begin{center}
\begin{tabular}{rrrrrrr}
& n & MSE & Incorr0(0) & Corr0(48) & PCD & Estimated Value\\\hline
POWL & 30 & 1.60 & 1.70 & 42.23 & 0.615(0.02) & 1.09(0.02)\tabularnewline
 & 100 & 1.27 & 1.94 & 46.64 & 0.768(0.02) & 1.27(0.02)\tabularnewline
 & 200 & 1.09 & 1.99 & 47.78 & 0.786(0.02) & 1.30(0.03)\tabularnewline
\hline 
PAL & 30 & 1.83 & 1.76 & 46.23 & 0.631(0.012) & 1.17(0.015) \tabularnewline
 & 100 & 1.01 & 0.92 & 46.53 & 0.808(0.009) & 1.37(0.009) \tabularnewline
 & 200 & 0.32 & 0.17 & 47.00 & 0.903(0.005)  & 1.45(0.006) \tabularnewline
\hline 
SCAL & 30 & 1.40 & 0.73 & 35.79 & 0.659(0.01) & 1.16(0.01)\tabularnewline
 & 100 & 0.52 & 0.11 & 41.97 & 0.764(0.01) & 1.31(0.01)\tabularnewline
 & 200 & 0.19 & 0.01 & 46.03 & 0.749(0.01) & 1.32(0.01)\tabularnewline
\hline 
SMCAL & 30 & 0.93 & 0.82 & 43.09 & 0.677(0.014) & 1.21(0.017) \tabularnewline
 & 100 & 0.81 & 0.75 & 44.11 & 0.735(0.010) & 1.28(0.013) \tabularnewline
 & 200 & 0.69 & 0.57 & 45.42 & 0.788(0.007) & 1.34(0.009) \tabularnewline
\hline 
SMCAL-Discrete & 30 & 0.95 & 0.89 & 43.12 & 0.653(0.014) & 1.20(0.017) \tabularnewline
 & 100 & 0.79 & 0.74 & 44.15 & 0.723(0.009) & 1.28(0.011) \tabularnewline
 & 200 & 0.70 & 0.50 & 43.78 & 0.764(0.006) & 1.33(0.008) \tabularnewline
\end{tabular}
\end{center}
\end{table}

In general, POWL gives many estimated zeros. Our MSEs are larger than SCAL but smaller than POWL. In fact, our theoretical results indicate a slower convergence rate than SCAL. In the case when the sample size is 200, our PCDs and Estimated Values are significantly higher than those in SCAL, which may be explained by SMCAL has a smaller bias compared to SCAL. We also
notice that the PCD of 200-sample SCAL is abnormally smaller than the PCD of 100-sample SCAL. This may be because SCAL does not converge to the true coefficients due to the induced bias by the hinge loss. In summary, we think when the sample size is small, SCAL is better than SMCAL. But when the sample size grows larger, SMCAL will be increasingly better than SCAL.

In the discrete cases, the MSE, variable selection, PCDs and Estimated Values have similar patterns. 

PAL has the best performance in this example. The 100-sample PAL is already much better than 200-sample SCAL and SMCAL in terms of PCDs and Estimated Values. In this example, the contrast function $D(X_i)=0.884(1-X_{i1}-X_{i2})$, which is a linear combination of the features. In such a linear case, pairwise comparison based methods like CAL, SCAL and SMCAL seems to be less efficient than PAL. But CAL, SCAL and SMCAL only assume $D(X_{i})=Q(\beta^{*T}X_{i})$, which is flexible and can be applied in nonlinear cases.

\subsubsection{Nonlinear Case}
\label{sssec:nonlinear}
The second example, see Table \ref{tab:nonlinear}, compares PAL and SMCAL on a case when the contrast function is nonlinear of the features. We set $d=50$, $n=30,100,200$, $X_i\sim \mathcal{N}(0,I)$, $\pi(X_i)=0.5$ and $Y_i=1+X_{i,1}-X_{i,2}+X_{i,3}+X_{i,4}+A\cdot (exp(1+X_{i,1}+X_{i,2}-X_{i,3}+X_{i,4})-exp(1))+\epsilon$, where $\epsilon\sim \mathcal{N}(0,0.5)$. Our MSEs are much smaller than PAL, and we can successfully select important features when the sample size is relatively large. Our PCDs and Estimated Values are also much better than PAL. PAL seems to be not suitable for such a nonlinear case, but SMCAL performs quite well.

\begin{table}
\caption{Comparison between PAL and SMCAL: nonlinear case.  \label{tab:nonlinear}}
\small
\begin{center}
\begin{tabular}{rrrrrrr}
& n & MSE & Incorr0(0) & Corr0(46) & PCD & Estimated Value\\\hline
PAL & 30 & 1.44 & 3.12 & 44.30 & 0.579(0.007) & 10.42(0.431) \tabularnewline
 & 100 & 0.82 & 2.22 & 45.35 & 0.666(0.006) & 14.67(0.486) \tabularnewline
 & 200 & 0.58 & 1.80 & 45.83 & 0.706(0.005) & 16.39(0.552) \tabularnewline
\hline 
SMCAL & 30 & 0.90 & 1.03 & 33.53 & 0.689(0.006) & 17.45(0.464) \tabularnewline
 & 100 & 0.52 & 0.05 & 29.44 & 0.814(0.006) & 18.56(0.319) \tabularnewline
 & 200 & 0.28 & 0.00 & 30.66 & 0.888(0.004) & 19.92(0.510) \tabularnewline
\end{tabular}
\end{center}
\end{table}

\subsection{High Dimension}
\label{ssec:high-dim}

We conduct simulations for the following six high dimensional models
which can be found in Section 4.2, \cite{scal}. We uniformly
set sample size $n=100$ and dimension $d=500$. The interaction terms
with treatment type $A$ are $X\beta$ or $(X\beta)^{3}$, which satisfies our Monotone Single Index Model assumption. However, the baseline function can be very flexible. The following models use linear, polynomial or even trigonometric sines as the baseline function. We set $X\sim \mathcal{N}(0,I)$, $\pi(X_{i})=0.5$ and $\varepsilon \sim \mathcal{N}(0,1)$.

$\mathbf{Model1}:$ $Y=3+X\gamma_{1}+X\beta A+\varepsilon$, where
$\gamma_{1}=(-1,1,\mathbf{0}_{d-2})^{T}$, $\beta=(2,1.8,0,0,0,-1.6,\mathbf{0}_{d-6})^{T}$.

$\mathbf{Model2}:$ $Y=3-0.5(X\gamma_{1})^{2}+0.625(X\gamma_{2})^{2}+X\beta A+\varepsilon$,
where $\gamma_{1}=(1,0.5,\mathbf{0}_{d-2})^{T}$, $\gamma_{2}=(0,1,\mathbf{0}_{d-2})^{T}$,
$\beta=(2,1.8,0,0,0,-1.6,\mathbf{0}_{d-6})^{T}$.

$\mathbf{Model3}:$ $Y=1-sin(X\gamma_{1})+sin(X\gamma_{2})+X\beta A+\varepsilon$,
where $\gamma_{1}=(1,\mathbf{0}_{d-1})^{T}$, $\gamma_{2}=(0,1,\mathbf{0}_{d-2})^{T}$,
$\beta=(2,1.8,0,0,0,-1.6,\mathbf{0}_{d-6})^{T}$.

$\mathbf{Model4}:$ $Y=3+X\gamma_{1}+(X\beta)^{3}A+\varepsilon$,
where $\gamma_{1}=(-1,1,\mathbf{0}_{d-2})^{T}$, $\beta=(1,0.9,0,0,0,-0.8,\mathbf{0}_{d-6})^{T}$.

$\mathbf{Model5}:$ $Y=3-0.5(X\gamma_{1})^{2}+0.625(X\gamma_{2})^{2}+(X\beta)^{3}A+\varepsilon$,
where $\gamma_{1}=(1,0.5,\mathbf{0}_{d-2})^{T}$, $\gamma_{2}=(0,1,\mathbf{0}_{d-2})^{T}$,
$\beta=(1,0.9,0,0,0,-0.8,\mathbf{0}_{d-6})^{T}$.

$\mathbf{Model6}:$ $Y=1-sin(X\gamma_{1})+sin(X\gamma_{2})+(X\beta)^{3}A+\varepsilon$,
where $\gamma_{1}=(1,\mathbf{0}_{d-1})^{T}$, $\gamma_{2}=(0,1,\mathbf{0}_{d-2})^{T}$,
$\beta=(1,0.9,0,0,0,-0.8,\mathbf{0}_{d-6})^{T}$.

\cite{scal} reported the results of SCAL, and we ran simulations for SMCAL, all summarized in Table \ref{tab:high}. 
\begin{table}
\scriptsize
\caption{Simulation results of SCAL and SMCAL: high dimensional cases.  \label{tab:high}}
\begin{center}
\begin{tabular}{rrrrrrrrrrr}
Model & \multicolumn{2}{c}{MSE} & \multicolumn{2}{c}{Incorr0(0)} & \multicolumn{2}{c}{Corr0(497)} & \multicolumn{2}{c}{PCD} & \multicolumn{2}{c}{Estimated Value}\tabularnewline
& SCAL & SMCAL & SCAL & SMCAL & SCAL & SMCAL & SCAL & SMCAL & SCAL & SMCAL\tabularnewline
\hline 
Model 1 & 0.61 & 0.76 & 0.75 & 0.86 & 482.62 & 490.81 & 0.744(0.01) & 0.732(0.005) & 3.80(0.02) & 3.82(0.016) \tabularnewline
Model 2 & 0.56 & 0.71 & 0.57 & 0.55 & 485.34 & 492.69 & 0.763(0.01) & 0.749(0.006) & 3.79(0.03) & 3.87(0.015) \tabularnewline
Model 3 & 0.44 & 0.69 & 0.49 & 0.47 & 488.12 & 491.59 & 0.786(0.01) & 0.751(0.005) & 1.92(0.02) & 1.88(0.015) \tabularnewline
Model 4 & 0.35 & 0.62 & 0.35 & 0.21 & 486.81 & 481.69 & 0.801(0.01) & 0.752(0.006) & 5.67(0.04) & 5.69(0.043) \tabularnewline
Model 5 & 0.32 & 0.56 & 0.25 & 0.18 & 487.00 & 485.41 & 0.810(0.01) & 0.740(0.008) & 5.63(0.05) & 5.62(0.057) \tabularnewline
Model 6 & 0.29 & 0.55 & 0.20 & 0.05 & 485.33 & 485.09 & 0.820(0.01) & 0.774(0.007) & 3.74(0.04) & 3.78(0.040) \tabularnewline
\end{tabular}
\end{center}
\end{table}
It shows that our model selections and Estimated Values are close to those in SCAL. The MSEs and PCDs are not as good as those in SCAL, perhaps because SCAL has a faster convergence rate, which makes it perform better than SMCAL in high dimensions. 

\section{Real Data Analysis}
\label{sec:real_data}

The STAR{*}D study, which focused on the Major Depression Disorder
(MDD), enrolled over 4,000 outpatients from age 18 to age 75. There
were altogether four treatment levels: 2, 2A, 3 and 4. At each treatment level, patients were randomly assigned to different treatment groups. Various clinic and socioeconomic factors were recorded, as well as the treatment outcomes. More details can be found in \cite{stard}. 

We focus on level 2 of STAR{*}D study. The data consists of 319 samples,
with each sample containing the patient ID, treatment type, treatment
outcome and 305 other clinical or genetic features. There are two
different treatment types: bupropion (BUP) and sertraline (SER). We
use $0$ to label SER and $1$ to label BUP. The 16-item Quick Inventory of Depressive Symptomatology-Clinician-Rated (QIDS-C16) ranges from 0 to 24, with smaller values indicating better outcomes. In our setting, larger value indicates better treatment effect, so we use the negative QIDS-C16 as our treatment outcome. 

Among the 319 patients selected, 166 of them have received SER and
153 of them have received BUP. After obtaining our estimated treatment
regime, we use the inverse probability weighted estimator (IPW) proposed
by \cite{ipwe} 
\[
\frac{1}{n}{\displaystyle \sum_{i=1}^{n}\frac{Y_{i}\mathbbm{1}[A_{i}=\mathbbm{1}(\hat{\beta^{T}X_{i}>c})]}{A_{i}\pi(X_{i})+(1-A_{i})[1-\pi(X_{i})]}}
\]
to calculate the estimated values. We draw bootstrap samples for 1000 times, each
time with 1000 samples to get the Estimated Value and 95\% CI. 

We ran the simulation for PAL and our SMCAL. \cite{scal} reported the results of SCAL, POWL and non-dynamic treatment regimes. These are all summarized in Table \ref{tab:est_value}. 
\begin{table}
\caption{Estimated Values of POWL, SCAL and SMCAL.  \label{tab:est_value}}
\begin{center}
\begin{tabular}{rrrr}
Treatment Regime & Estimated Value & Diff & 95\% CI on Diff\tabularnewline
\hline 
Optimal Regime (SCAL) & -6.77 &  & \tabularnewline
Optimal Regime (SMCAL) & -6.90 & 0.13 & (-0.46,0.69) \tabularnewline
Optimal Regime (PAL) & -8.15 & 1.38 & (0.71,2.02)\tabularnewline
Optimal Regime (POWL) & -9.46 & 2.69 & (1.18,4.24)\tabularnewline
BUP & -10.50 & 3.75 & (2.38,5.19)\tabularnewline
SER & -10.72 & 3.97 & (2.57,5.50)\tabularnewline
\end{tabular}
\end{center}
\end{table}
Consider SCAL as baseline, \emph{Diff} in the table means the difference between the estimated value of each method and the estimated value of SCAL. According to Table \ref{tab:est_value}, the optimal treatment regimes are all better than non-dynamic treatment regimes. Optimal treatment regimes by SCAL and SMCAL have the best estimated values. PAL and POWL are not as good as SCAL and SMCAL in this real data example.

Comparison of the received treatments and the estimated optimal treatment regimes are available in Table \ref{tab:treatment-smcal}. 
\begin{table}
\small
\caption{Treatment recommended by SMCAL. \label{tab:treatment-smcal}}
\begin{center}
\begin{tabular}{r|rr}
 & Recommended: SER & Recommended: BUP\tabularnewline
\hline 
Randomized Treatment: SER & 75 & 91 \tabularnewline
Randomized Treatment: BUP & 68 & 85 \tabularnewline
\end{tabular}
\end{center}
\end{table}
\emph{Randomized Treatment} means the treatment actually performed in the STAR*D study. \emph{Estimated Treatment} means the treatment suggested by the estimated optimal treatment regime. From Table \ref{tab:treatment-smcal} we can see that among those 153 people who received BUP, the optimal treatment regime by SMCAL suggests that 77 of them receive SER and 76 of them still receive BUP. It is a rather balanced treatment regime.

\section{Conclusions}
\label{sec:conc}
In this paper, we have proposed SMCAL, which is based on the concordance-assisted-learning (CAL) framework and the smoothing procedure by \cite{hanfang}, and aims at dealing with the optimization issue of CAL in high dimensional data. We established convergence results in both continuous covariates and discrete covariates cases. The proposed SMCAL can be successfully applied in the case when the contrast function is nonlinear of the features, and it does not induce a relatively large bias compared to SCAL.

\bigskip
\begin{center}
{\large\bf SUPPLEMENTARY MATERIAL}
\end{center}


\section{Proofs}

We first present Lemma \ref{lem:D_bound} which will be used in the proof of Lemma \ref{lem:stochastic_error_rate}. 

\begin{lemma}
\label{lem:D_bound}
There exists $M_{5}$, s.t. ${\displaystyle {\sup_{h\in \mathbb{R}^{d},\|h\|_{0}\leq4s_{n}}\frac{h^{T}D^{T}Dh}{\frac{n(n-1)}{2}\|h\|_{2}^{2}}}\leq M_{5}s_n}$,
where 
\begin{gather*}
D=\begin{pmatrix}X_{1}^{T}-X_{2}^{T}\\
X_{1}^{T}-X_{3}^{T}\\
...\\
X_{1}^{T}-X_{n}^{T}\\
X_{2}^{T}-X_{3}^{T}\\
...
\end{pmatrix}_{\frac{n(n-1)}{2}\times d}
\end{gather*}
\end{lemma}

\begin{proof}[Proof of Lemma \ref{lem:D_bound}]

(A5) in Assumptions \ref{asm:stochastic} or (B3) in Assumption \ref{asm:discrete} implies that 
$${\displaystyle \sup_{h\in \mathbb{R}^{d}:\|h\|_{0}\leq4s_{n},\|h\|_{2}=1}((X_{i}-X_{j})^{T}h)^{2}}
\leq 4 M_{4}^{2}{\displaystyle \sup_{h\in \mathbb{R}^{d}:\|h\|_{0}\leq4s_{n},\|h\|_{2}=1}\|h\|_{1}^{2}} \leq 16 M_{4}^{2}s_{n}.$$
Therefore,
\begin{align}
\sup_{h\in \mathbb{R}^{d}:\|h\|_{0}\leq4s_{n}}\frac{h^{T}D^{T}D h}{\frac{n(n-1)}{2}\|h\|_{2}^{2}} & \leq\frac{2}{n(n-1)}{\displaystyle {\displaystyle \sup_{h\in \mathbb{R}^{d}:\|h\|_{0}\leq4s_{n}}\frac{h^{T}({\displaystyle \sum_{1\leq i<j\leq n}}(X_{i}-X_{j})(X_{i}-X_{j})^{T})h}{\|h\|_{2}^{2}}}}\nonumber \\
 & \leq\frac{2}{n(n-1)}{{\displaystyle \sum_{1\leq i<j\leq n}}{\displaystyle \sup_{h\in \mathbb{R}^{d}:\|h\|_{0}\leq4s_{n}}\frac{((X_{i}-X_{j})^{T}h)^{2}}{\|h\|_{2}^{2}}}}\nonumber \\
 & <{\displaystyle {\sup_{h\in \mathbb{R}^{d}:\|h\|_{0}\leq4s_{n},\|h\|_{2}=1}((X_{i}-X_{j})^{T}h)^{2}}}\leq16M_{4}^{2}s_{n},
\end{align}
This implies that if we let $M_{5}=16M_{4}^{2}$, then
\[
{\displaystyle {\sup_{h\in \mathbb{R}^{d},\|h\|_{0}\leq4s_{n}}}\frac{h^{T}D^{T}Dh}{\frac{n(n-1)}{2}\|h\|_{2}^{2}}}\leq M_{5}s_{n}.
\]
\end{proof}

\subsection{Proofs of Section \ref{ssec:cts}}
\begin{proof}[Proof of Lemma \ref{lem:c_beta_approx_rate}]
Define 
$$W(v,\beta)=E[(Q(X_{i}^{T}\beta^{*})-Q(X_{j}^{T}\beta^{*}))|(X_{i}^{T}-X_{j}^{T})\beta=v]f_{(X_{i}^{T}-X_{j}^{T})\beta}(v),$$
and
\[T=\frac{1}{2}(X_{i}^{T}+X_{j}^{T})\beta,\quad U=X^{T}\beta^{*}-X^{T}\beta,\] then
\begin{align}
 & W(v,\beta) \nonumber \\
= & \int_{\mathbb{R}}\int_{\mathbb{R}}\int_{\mathbb{R}}(Q(t+\frac{v}{2}+u_{1})-Q(t-\frac{v}{2}+u_{2}))f_{(X^{T}\beta,U)}(t+\frac{v}{2},u_{1})f_{(X^{T}\beta,U)}(t-\frac{v}{2},u_{2})du_{1}du_{2}dt\nonumber \\
= & \int_{\mathbb{R}}\int_{\mathbb{R}}Q(t+\frac{v}{2}+u_{1})f_{(X^{T}\beta,U)}(t+\frac{v}{2},u_{1})f_{X^{T}\beta}(t-\frac{v}{2})du_{1}dt\nonumber \\
&-\int_{\mathbb{R}}\int_{\mathbb{R}}Q(t-\frac{v}{2}+u_{2})f_{X^{T}\beta}(t+\frac{v}{2})f_{(X^{T}\beta,U)}(t-\frac{v}{2},u_{2})du_{2}dt\nonumber \\
= & \int_{\mathbb{R}}\int_{\mathbb{R}}Q(t+\frac{v}{2}+u_{1})f_{(X^{T}\beta,U)}(t+\frac{v}{2},u_{1})f_{X^{T}\beta}(t-\frac{v}{2})du_{1}dt\nonumber \\
 & -\int_{\mathbb{R}}\int_{\mathbb{R}}Q(t+\frac{v}{2}+u_{1})f_{X^{T}\beta}(t+\frac{3v}{2})f_{(X^{T}\beta,U)}(t+\frac{v}{2},u_{1})du_{1}dt\nonumber \\
= & \int_{\mathbb{R}}\int_{\mathbb{R}}Q(t+\frac{v}{2}+u_{1})f_{(X^{T}\beta,U)}(t+\frac{v}{2},u_{1})(f_{X^{T}\beta}(t-\frac{v}{2})-f_{X^{T}\beta}(t+\frac{3v}{2}))du_{1}dt
\end{align}

(A1) and (A2) in Assumptions \ref{asm:approx} implies $|Q(t+\frac{v}{2}+u_{1})|\leq M_{1}$ and
$|f_{X^{T}\beta}(t-\frac{v}{2})-f_{X^{T}\beta}(t+\frac{3v}{2})|\leq2|v|M_{3}$.
So
$$|W(v,\beta)|\leq2M_{1}M_{3}|v|.$$

Notice that
\begin{equation}
\begin{aligned} 
& E[(Q(X_{i}^{T}\beta^{*})-Q(X_{j}^{T}\beta^{*}))(1-F((X_{i}^{T}-X_{j}^{T})\beta\alpha_{n}))\mathbbm{1}((X_{i}^{T}-X_{j}^{T})\beta>0)]\nonumber\\
= & E[(Q(X_{i}^{T}\beta^{*})-Q(X_{j}^{T}\beta^{*}))F((X_{j}^{T}-X_{i}^{T})\beta\alpha_{n})\mathbbm{1}((X_{j}^{T}-X_{i}^{T})\beta<0)]\nonumber\\
= & E[(Q(X_{j}^{T}\beta^{*})-Q(X_{i}^{T}\beta^{*}))F((X_{i}^{T}-X_{j}^{T})\beta\alpha_{n})\mathbbm{1}((X_{i}^{T}-X_{j}^{T})\beta<0)],
\end{aligned}
\end{equation}
we have
\begin{equation}
\begin{aligned}\begin{aligned}\end{aligned}
|C(\beta)-C_{n}(\beta)|= & |E[(Q(X_{i}^{T}\beta^{*})-Q(X_{j}^{T}\beta^{*}))(\mathbbm{1}((X_{i}^{T}-X_{j}^{T})\beta>0)-F((X_{i}^{T}-X_{j}^{T})\beta\alpha_{n}))]|\\
= & |E[(Q(X_{i}^{T}\beta^{*})-Q(X_{j}^{T}\beta^{*}))(1-F((X_{i}^{T}-X_{j}^{T})\beta\alpha_{n}))\mathbbm{1}((X_{i}^{T}-X_{j}^{T})\beta>0)\nonumber\\
 & -(Q(X_{i}^{T}\beta^{*})-Q(X_{j}^{T}\beta^{*}))F((X_{i}^{T}-X_{j}^{T})\beta\alpha_{n})\mathbbm{1}((X_{i}^{T}-X_{j}^{T})\beta<0)]|\\
= & 2|E[(Q(X_{j}^{T}\beta^{*})-Q(X_{i}^{T}\beta^{*}))F((X_{i}^{T}-X_{j}^{T})\beta\alpha_{n})\mathbbm{1}((X_{i}^{T}-X_{j}^{T})\beta<0)]|.
\end{aligned}
\end{equation}
Therefore,
\begin{equation}
\begin{aligned}
&\sup_{\beta\in\mathbb{R}^{d}:\beta_1=1,\|\beta^{*}-\beta\|_{2}\leq r(\gamma)}|C(\beta)-C_{n}(\beta)|\\
=&  2\sup_{\beta\in\mathbb{R}^{d}:\beta_1=1,\|\beta^{*}-\beta\|_{2}\leq r(\gamma)}|\int_{-\infty}^{0}W(v,\beta)\frac{1}{1+e^{-\alpha_{n}v}}d v|\\
\leq & 4M_{1}M_{3}|\int_{0}^{\infty}\frac{v}{1+e^{\alpha_{n}v}}dv|\\
=& 4M_{1}M_{3}\alpha_{n}^{-2}|\int_{0}^{\infty}\frac{v}{1+e^{v}}dv|\lesssim\alpha_{n}^{-2}.
\end{aligned}
\end{equation}
\end{proof}

\begin{proof}[Proof of Theorem \ref{thm:approx_rate}]
Notice that
\[
C(\beta_{r(\gamma),\alpha_{n}}^{*})-C_{n}(\beta_{r(\gamma),\alpha_{n}}^{*})\leq C(\beta_{r(\gamma),\alpha_{n}}^{*})-C_{n}(\beta^{*})\leq C(\beta^{*})-C_{n}(\beta^{*}),
\]
we have 
\begin{equation}
\begin{aligned}
|C(\beta^{*})-C(\beta_{r(\gamma),\alpha_{n}}^{*})|\leq&|C_{n}(\beta^{*})-C(\beta_{r(\gamma),\alpha_{n}}^{*})|+|C(\beta^{*})-C_{n}(\beta^{*})|\\
\leq& 2\sup_{\beta\in\mathbb{R}^{d}:\beta_1=1,\|\beta^{*}-\beta\|_{2}\leq r(\gamma)}|C(\beta)-C_{n}(\beta)|\lesssim \alpha_{n}^{-2}.
\end{aligned}
\end{equation}
The last step is by Lemma \ref{lem:c_beta_approx_rate}.

Similar to \cite{hanfang}, we can pick a positive constant set $\gamma=\{\gamma_{1},\gamma_{2}\}$ with $\dfrac{\gamma_{2}}{\gamma_{1}}$
close to $1$, such that for some small enough $r(\gamma)>0$, $\beta_1=\beta_1^*=1$ and $\parallel\beta-\beta^{*}\parallel_{2}$$\leq r(\gamma)$, we have
\begin{equation}
\gamma_{1}\lambda_{min}(\Gamma)\parallel\beta-\beta^{*}\parallel_{2}^{2}\leq C(\beta^{*})-C(\beta)\leq\gamma_{2}\lambda_{max}(\Gamma)\parallel\beta-\beta^{*}\parallel_{2}^{2},
\end{equation}
then by (A3) in Assumption \ref{asm:approx}, we know $\lambda_{min}(\Gamma)\geq c_3$, so $$\parallel\beta_{r(\gamma),\alpha_{n}}^{*}-\beta^{*}\parallel_{2}^{2}\leq \frac{C(\beta^{*})-C(\beta_{r(\gamma),\alpha_{n}}^{*})}{\gamma_{1}\lambda_{min}(\Gamma)} \lesssim\alpha_{n}^{-2}.$$
\end{proof}

\begin{proof}[Proof of Lemma \ref{lem:stochastic_inf_norm}]

Let $$h_{k}(x_{i},x_{j})=\frac{\partial}{\partial\beta_{k}}(w_{i}-w_{j})F((x_{i}-x_{j})^{T}\beta\alpha_{n})|_{\beta=\beta_{r(\gamma),\alpha_{n}}^{*}},$$
where $k=2,3,...,d$. Then $Eh_{k}(X_{i},X_{j})=0$ since $\beta_{r(\gamma),\alpha_{n}}^{*}$ is the locally maximizer of $C_n(\beta)$. Because 
$$F^{'}(x)=1/(e^{x}+e^{-x}+2)\leq\frac{1}{4},$$ 
under
Assumption \ref{asm:stochastic} or Assumption \ref{asm:discrete}, we both have

$\begin{aligned}
|h_{k}(X_{i},X_{j})|= & |(W_{i}-W_{j})F^{'}((X_{i}-X_{j})^{T}\beta_{r(\gamma),\alpha_{n}}^{*}\alpha_{n})\alpha_{n}(X_{i,k}-X_{j,k})|\\
\text{\ensuremath{\leq}} & \frac{1}{4}\alpha_{n}|(W_{i}-W_{j})(X_{i,k}-X_{j,k})|\leq\alpha_{n}M_{2}M_{4}.
\end{aligned}$

Therefore, using the Hoeffding bound for U-statistics which can be found in, for example, \cite{pitcan2017},
\begin{align}
P\{\|\nabla\hat{L}_{n}(\beta_{r(\gamma),\alpha_{n}}^{*})\|_{\infty}\geq t\} & =P\{\|\nabla\hat{C}_{n}(\beta_{r(\gamma),\alpha_{n}}^{*})\|_{\infty}\geq t\}\nonumber \\
\leq & d\cdot P\{|{\displaystyle \sum_{1\leq i<j\leq n}\frac{2}{n(n-1)}h_{k}(x_{i},x_{j})|\geq t\}}\nonumber \\
\leq & d\cdot2exp\{-\frac{\lfloor n/2\rfloor t^{2}}{2\alpha_{n}^{2}M_{2}^{2}M_{4}^{2}}\}
\end{align}
Choose $t=\sqrt{8}\alpha_{n}M_{2}M_{4}\sqrt{\frac{log(d)}{n}}$,
the proof is completed.
\end{proof}

\begin{proof}[Proof of Lemma \ref{lem:stochastic_error_rate}]
Define
\begin{align*}
B(h) =&|\hat{L}_{n}(\beta_{r(\gamma),\alpha_{n}}^{*}-h)-L_{n}(\beta_{r(\gamma),\alpha_{n}}^{*}-h)-\hat{L}_{n}(\beta_{r(\gamma),\alpha_{n}}^{*})+L(\beta_{r(\gamma),\alpha_{n}}^{*})|\\
 =&|\frac{1}{n(n-1)}{\sum_{i\neq j}}(w_{i}-w_{j})F((x_{i}-x_{j})^{T}(\beta_{r(\gamma),\alpha_{n}}^{*}-h)\alpha_{n})\\
 &- E[(W_{i}-W_{j})F((X_{i}-X_{j})^{T}(\beta_{r(\gamma),\alpha_{n}}^{*}-h)\alpha_{n})]\\
 &-\frac{1}{n(n-1)} {\sum_{i\neq j}}(w_{i}-w_{j})F((x_{i}-x_{j})^{T}\beta_{r(\gamma),\alpha_{n}}^{*}\alpha_{n})\\
 &+E[(W_{i}-W_{j})F((X_{i}-X_{j})^{T}\beta_{r(\gamma),\alpha_{n}}^{*}\alpha_{n})]|
\end{align*}

Under Assumption \ref{asm:stochastic}, if $\|h\|_0\leq 2s_n$,
\begin{align*}
 & |(w_{i}-w_{j})F((X_{i}^{T}-X_{j}^{T})(\beta_{r(\gamma),\alpha_{n}}^{*}-h)\alpha_{n})-(w_{i}-w_{j})F((X_{i}^{T}-X_{j}^{T})\beta_{r(\gamma),\alpha_{n}}^{*}\alpha_{n})|\\
\leq & \frac{1}{4}|w_{i}-w_{j}|\cdot|(X_{i}^{T}-X_{j}^{T})h|\cdot\alpha_{n}\leq M_{2}M_{4}\alpha_{n}\|h\|_{1}\leq M_{2}M_{4}\alpha_{n}\sqrt{2 s_{n}}\|h\|_{2}.
\end{align*}
According to the Hoeffding bound for U-statistics, 
$$
P(B(h)\geq t)\stackrel{p}{\leq}2\exp\{-\frac{\lfloor n/2\rfloor t^{2}}{4 M_{2}^{2}M_{4}^{2}\alpha_{n}^{2}s_{n}\|h\|_{2}^{2}}\}.
$$
Let $t=\frac{\|h\|_{2}t^{'}}{\sqrt{n}}$ and $M_{6}=8M_{2}^{2}M_{4}^{2}$,
then
\[
P(\frac{B(h)}{\|h\|_{2}}\geq\frac{t^{'}}{\sqrt{n}})\stackrel{p}{\leq}2\exp\{-\frac{t^{'2}}{s_{n}\alpha_{n}^{2}M_{6}}\}.
\]
Choose $t^{'}=C\sqrt{2s_{n}(s_{n}+1)log(d)}\alpha_{n}$, then
\begin{equation}
P(\frac{B(h)}{\|h\|_{2}}\geq\frac{C\sqrt{2s_{n}(s_{n}+1)log(d)}\alpha_{n}}{\sqrt{n}})\stackrel{p}{\leq}2d^{-\frac{2(s_{n}+1)C^{2}}{M_{6}}}.
\end{equation}

Let $N_{1}$ be an $\epsilon$-covering of $\{h\in \mathbb{R}^{d}:\|h\|_{0}\leq2s_{n},\:\|h\|_{2}= 2r(\gamma)\}$, where we constrain $N_1\subset \{h\in \mathbb{R}^{d}:\|h\|_{0}\leq2s_{n},\:\|h\|_{2}= 2r(\gamma)\}$ and $\epsilon$ is a small positive constant to be chosen. Since the covering number of $\{u\in \mathbb{R}^{2s_{n}}:\:\|u\|_{2}= 2r(\gamma)\}$ should not exceed the packing number, which is less than
$$(\frac{1+\frac{\epsilon}{2r(\gamma)}}{\frac{\epsilon}{2r(\gamma)}})^{2s_n},$$
we can find a $N_1$ such that
$$|N_1|\leq {d \choose 2s_{n}}(\frac{4r(\gamma)}{\epsilon})^{2s_{n}}.$$ 
Consider 
\begin{align}
N={\displaystyle \bigcup_{i=1}^{\lceil2r(\gamma)/\epsilon\rceil}\frac{i}{\lceil2r(\gamma)/\epsilon\rceil}\cdot N_{1}},\label{equ:e-net_def}
\end{align}
then
\begin{align}
\label{equ:e-net}
 & P({\displaystyle {\sup_{h\in N}}\frac{B(h)}{\|h\|_{2}}\geq\frac{C\sqrt{2s_{n}(s_{n}+1)log(d)}\alpha_{n}}{\sqrt{n}})}\nonumber \\
\stackrel{p}{\leq} & \frac{4r(\gamma)}{\epsilon}{d \choose 2s_{n}}(\frac{4r(\gamma)}{\epsilon})^{2s_{n}}2d^{-\frac{2(s_{n}+1)C^{2}}{M_{6}}}\leq2(\frac{4r(\gamma)}{\epsilon}d^{1-\frac{C^{2}}{M_{6}}})^{2s_{n}+1}.
\end{align}

When $\|h_{1}\|_{2}=\|h_{2}\|_{2}$,
\begin{align}
& |\frac{B(h_{1})}{\|h_{1}\|_{2}}-\frac{B(h_{2})}{\|h_{2}\|_{2}}|=\frac{1}{\|h_1\|_2}|B(h_{1})-B(h_{2})| \nonumber \\
= & \frac{1}{\|h_1\|_2}|\frac{1}{n(n-1)}{\sum_{i\neq j}}(w_{i}-w_{j})F((x_{i}-x_{j})^{T}(\beta_{r(\gamma),\alpha_{n}}^{*}-h_1)\alpha_{n})\nonumber\\
&-{\displaystyle E[(W_{i}-W_{j})F((X_{i}-X_{j})^{T}(\beta_{r(\gamma),\alpha_{n}}^{*}-h_1)\alpha_{n})]}\nonumber\\
&-\frac{1}{n(n-1)}{\sum_{i\neq j}}(w_{i}-w_{j})F((x_{i}-x_{j})^{T}(\beta_{r(\gamma),\alpha_{n}}^{*}-h_2)\alpha_{n})\nonumber\\
&+{\displaystyle E[(W_{i}-W_{j})F((X_{i}-X_{j})^{T}(\beta_{r(\gamma),\alpha_{n}}^{*}-h_2)\alpha_{n})]}|\nonumber\\
\leq & \frac{M_{2}\alpha_{n}}{2n(n-1)}\|D(\frac{h_{1}}{\|h_{1}\|_{2}}-\frac{h_{2}}{\|h_{2}\|_{2}})\|_{1}+\frac{M_{2}\alpha_{n}}{2}\cdot E|(X_{i}^{T}-X_{j}^{T})(\frac{h_{1}}{\|h_{1}\|_{2}}-\frac{h_{2}}{\|h_{2}\|_{2}})|\nonumber \\
\leq & \frac{\frac{1}{4}M_{2}\alpha_{n}}{\sqrt{\frac{n(n-1)}{2}}}\|D(\frac{h_{1}}{\|h_{1}\|_{2}}-\frac{h_{2}}{\|h_{2}\|_{2}})\|_{2}+M_{2}\alpha_{n}\cdot E|X_{i}^{T}(\frac{h_{1}}{\|h_{1}\|_{2}}-\frac{h_{2}}{\|h_{2}\|_{2}})|,
\end{align}
where $D$ is defined in Lemma \ref{lem:D_bound}. Then by Lemma \ref{lem:D_bound}, we have
\begin{align}
\label{equ:same_norm}
& {\displaystyle {\sup_{\|h_{1}-h_{2}\|_{0}\leq4s_{n},\|h_{1}\|_{2}=\|h_{2}\|_{2},\|\frac{h_{1}}{\|h_{1}\|_{2}}-\frac{h_{2}}{\|h_{2}\|_{2}}\|_{2}\leq\epsilon}}|\frac{B(h_{1})}{\|h_{1}\|_{2}}-\frac{B(h_{2})}{\|h_{2}\|_{2}}|}\nonumber \\
\leq &\frac{1}{4}M_{2}\alpha_{n}\epsilon\cdot(\sqrt{M_{5}s_{n}}+4E[\sup_{h\epsilon\mathbb{R}^{d}:\|h\|_{0}\leq4s_{n},\|h\|_{2}=1}|X_{i}^{T}h|])\nonumber \\
\leq & \frac{1}{4}M_{2}\alpha_{n}\epsilon\cdot(\sqrt{M_{5}s_{n}}+8M_{4}\sqrt{s_{n}}).
\end{align}
It bounds the difference between $\frac{B(h_{1})}{\|h_{1}\|_{2}}$
and $\frac{B(h_{2})}{\|h_{2}\|_{2}}$ when $\|h_{1}\|_{2}=\|h_{2}\|_{2}$. Next, consider $h_{1}$ and $h_{2}$ in the same direction but $\|h_{1}\|_{2}\neq \|h_{2}\|_{2}$. If $\|h\|_{2}=1$ and $\|h\|_0\leq 2s_n$, we have
\begin{align*}
 & |(\frac{1}{t}[F((x_{i}-x_{j})^{T}(\beta_{r(\gamma),\alpha_{n}}^{*}-t h)\alpha_{n})-F((x_{i}-x_{j})^{T}\beta_{r(\gamma),\alpha_{n}}^{*}\alpha_{n})])^{'}|\\
= & |\frac{-1}{t}F^{'}((x_{i}-x_{j})^{T}(\beta_{r(\gamma),\alpha_{n}}^{*}-t h)\alpha_{n})\alpha_{n}(x_{i}-x_{j})^{T}h\\
&-\frac{1}{t^2}[F((x_{i}-x_{j})^{T}(\beta_{r(\gamma),\alpha_{n}}^{*}-t h)\alpha_{n})-F((x_{i}-x_{j})^{T}\beta_{r(\gamma),\alpha_{n}}^{*}\alpha_{n})]|\\
= & |\frac{-1}{t}F^{'}((x_{i}-x_{j})^{T}(\beta_{r(\gamma),\alpha_{n}}^{*}-t h)\alpha_{n})\alpha_{n}(x_{i}-x_{j})^{T}h\\
&-\frac{1}{t^2}[F^{'}((x_{i}-x_{j})^{T}(\beta_{r(\gamma),\alpha_{n}}^{*}-t h)\alpha_{n})(-t(x_{i}-x_{j})^{T}h\alpha_{n})\\
&+\frac{1}{2}F^{''}((x_{i}-x_{j})^{T}(\beta_{r(\gamma),\alpha_{n}}^{*}-\xi h)\alpha_{n})\cdot(-t(x_{i}-x_{j})^{T}h\alpha_{n})^{2}]|\\
= & |\frac{1}{2}F^{''}((x_{i}-x_{j})^{T}(\beta_{r(\gamma),\alpha_{n}}^{*}-\xi h)\alpha_{n})\cdot\alpha_{n}^{2}((x_{i}-x_{j})^{T}h)^{2}|\leq c_{6}\alpha_{n}^{2}M_4^2\|h\|_1^2\leq 2c_{6}\alpha_{n}^{2}M_4^2 s_n,
\end{align*}
therefore,

\begin{align}
\label{equ:same_direct}
&{\displaystyle |\frac{B(t_{1}h)}{t_{1}}-\frac{B(t_{2}h)}{t_{2}}|}\nonumber\\
\leq & \frac{2M_{2}}{n(n-1)}{\sum_{i\neq j}}|[F((x_{i}-x_{j})^{T}(\beta_{r(\gamma),\alpha_{n}}^{*}-t_{1}h)\alpha_{n})-F((x_{i}-x_{j})^{T}\beta_{r(\gamma),\alpha_{n}}^{*}\alpha_{n})]/t_{1}\nonumber \\
 & -[F((x_{i}-x_{j})^{T}(\beta_{r(\gamma),\alpha_{n}}^{*}-t_{2}h)\alpha_{n})-F((x_{i}-x_{j})^{T}\beta_{r(\gamma),\alpha_{n}}^{*}\alpha_{n})]/t_{2}|\nonumber \\
 & +2M_{2}\cdot E|[F((X_{i}-X_{j})^{T}(\beta_{r(\gamma),\alpha_{n}}^{*}-t_{1}h)\alpha_{n})-F((X_{i}-X_{j})^{T}\beta_{r(\gamma),\alpha_{n}}^{*}\alpha_{n})]/t_{1}\nonumber \\
 & -[F((X_{i}-X_{j})^{T}(\beta_{r(\gamma),\alpha_{n}}^{*}-t_{2}h)\alpha_{n})-F((X_{i}-X_{j})^{T}\beta_{r(\gamma),\alpha_{n}}^{*}\alpha_{n})]/t_{2}|
\leq 8c_{6}M_{2}M_4^2 s_n\alpha_{n}^{2}|t_{1}-t_{2}|
\end{align}

For any $h\in \{h:\|h\|_{0}\leq2s_{n},\|h\|_{2}\leq2r(\gamma)\}$, we can find an integer $i\in \{1,2,...,\lceil2r(\gamma)/\epsilon\rceil\}$ and $h_0\in N$ s.t. $|\|h\|_{2}-\frac{2i r(\gamma)}{\lceil2r(\gamma)/\epsilon\rceil}|\leq \epsilon$, $\|h_0\|_2=\frac{2i r(\gamma)}{\lceil2r(\gamma)/\epsilon\rceil}$ and $\|h_0-\frac{2i r(\gamma)h}{\lceil2r(\gamma)/\epsilon\rceil\|h\|_2}\|_2\leq \epsilon$, where $N$ is defined in (\ref{equ:e-net_def}). 

Since $$|\frac{B(h)}{\|h\|_{2}}-\frac{B(h_0)}{\|h_0\|_{2}}|\leq |\frac{B(h)}{\|h\|_{2}}-\frac{B(\frac{2i r(\gamma)h}{\lceil2r(\gamma)/\epsilon\rceil\|h\|_2})}{\|\frac{2i r(\gamma)h}{\lceil2r(\gamma)/\epsilon\rceil\|h\|_2}\|_{2}}|+|\frac{B(\frac{2i r(\gamma)h}{\lceil2r(\gamma)/\epsilon\rceil\|h\|_2})}{\|\frac{2i r(\gamma)h}{\lceil2r(\gamma)/\epsilon\rceil\|h\|_2}\|_{2}}-\frac{B(h_0)}{\|h_0\|_{2}}|,$$
by (\ref{equ:same_norm}) and (\ref{equ:same_direct}) we have
\[
{\displaystyle {\sup_{\|h\|_{0}\leq2s_{n},\|h\|_{2}\leq2r(\gamma)}}\frac{B(h)}{\|h\|_{2}}}\leq {\displaystyle {\sup_{h\in N}}\frac{B(h)}{\|h\|_{2}}}+c_{7}M_{2}\alpha_{n}^{2}\epsilon,
\]
where $c_{7}=\frac{1}{4}\sqrt{M_{5}s_{n}}+2M_{4}\sqrt{s_{n}}+8M_4^2 s_n c_{6}$.

Let $\epsilon=\frac{\sqrt{2s_n(s_{n}+1)log(d)}}{\sqrt{n}c_{7}\alpha_{n}}$ and using (\ref{equ:e-net}), we get
\begin{align}
 & P({\displaystyle {\sup_{\|h\|_{0}\leq2s_{n},\|h\|_{2}\leq2r(\gamma)}}\frac{B(h)}{\|h\|_{2}}}\geq\frac{(C-1)\sqrt{2s_{n}(s_{n}+1)log(d)}M_{2}\alpha_{n}}{\sqrt{n}})\nonumber \\
\leq & 2(\frac{4r(\gamma)}{\epsilon}d^{1-\frac{C^{2}}{2M_{6}}})^{2s_{n}+1}=2(\frac{\sqrt{n}c_{7}\alpha_{n}4r(\gamma)}{\sqrt{2s_n(s_{n}+1)log(d)}}d^{1-\frac{C^{2}}{2M_{6}}})^{2s_{n}+1},
\end{align}
and we will let $\alpha_{n}\ll n$. Choose a large
enough $C$, we have 
\[
{\displaystyle {\sup_{\|h\|_{0}\leq2s_{n},\|h\|_{2}\leq2r(\gamma)}}\frac{B(h)}{\|h\|_{2}}\stackrel{p}{\lesssim}\frac{\sqrt{s_{n}^{2}log(d)}}{\sqrt{n}}\alpha_{n}}.
\]

The proof for the discrete cases under Assumption \ref{asm:discrete} is essentially the same.
\end{proof}

\begin{proof}[Proof of Theorem \ref{thm:stochastic_rate}]

The proof of Theorem \ref{thm:stochastic_rate} is very similar to the proof of Lemma 4.10 in \cite{hanfang}, so in here we only give a sketch of the proof. According to Lemma \ref{lem:c_beta_approx_rate}, 
\begin{equation}
L_{n}(\beta_{r(\gamma),\alpha_{n}}^{*}+\hat{\Delta})-L_{n}(\beta_{r(\gamma),\alpha_{n}}^{*})\geq  L(\beta_{r(\gamma),\alpha_{n}}^{*}+\hat{\Delta})-L(\beta_{r(\gamma),\alpha_{n}}^{*})-c_{5}\alpha_{n}^{-2}.
\end{equation}
And based on Taylor Expansion of $L(\beta_{r(\gamma),\alpha_{n}}^{*})-L(\beta^{*})$
and $L(\beta_{r(\gamma),\alpha_{n}}^{*}+\hat{\Delta})-L(\beta^{*})$,
\begin{equation}
    L(\beta_{r(\gamma),\alpha_{n}}^{*}+\hat{\Delta})-L(\beta_{r(\gamma),\alpha_{n}}^{*})\geq \frac{1}{2}\hat{\Delta}^{T}\Gamma\hat{\Delta}(1+o(1))-\hat{\Delta}^{T}\Gamma(\beta^{*}-\beta_{r(\gamma),\alpha_{n}}^{*})(1+o(1)),
\end{equation}
in here $\hat{\Delta}$ and $\beta^{*}-\beta_{r(\gamma),\alpha_{n}}^{*}$ do not include the first coordinate. Then use Theorem \ref{thm:approx_rate} and Assumption \ref{asm:approx}, we have
\begin{equation}
\begin{aligned}
\label{equ:expansion}
L_{n}(\beta_{r(\gamma),\alpha_{n}}^{*}+\hat{\Delta})-L_{n}(\beta_{r(\gamma),\alpha_{n}}^{*})
\geq & c_{8}\|\hat{\Delta}\|_{2}^{2}-c_{9}\alpha_{n}^{-1}\|\hat{\Delta}\|_{2}-c_{10}\alpha_{n}^{-2}
\end{aligned}
\end{equation}
where $c_{8},c_{9},c_{10}$ are some proper constants.

Similar to the proof of Theorem 4.8 in \cite{hanfang}, as long as $\lambda_{n}\geq2P^{*}(\nabla\hat{L}_{n}(\beta_{r(\gamma),\alpha_{n}}^{*}))$
and $\hat{L}_{n}(\cdot)$ is locally convex differentiable, where $P^{*}(\cdot)$
is the dual norm of $P(\cdot)$, we have 
\[
\|\hat{\Delta}_{S_{n}^{c}}\|_{1}\leq3\|\hat{\Delta}_{S_{n}}\|_{1}
\]
Then by Lemma \ref{lem:stochastic_inf_norm}, Lemma \ref{lem:stochastic_error_rate} and (\ref{equ:expansion}), define
\begin{equation}
\begin{aligned} & \delta\hat{L}_{n}(\hat{\Delta},\beta_{r(\gamma),\alpha_{n}}^{*}):=\hat{L}_{n}(\beta_{r(\gamma),\alpha_{n}}^{*}+\hat{\Delta})-\hat{L}_{n}(\beta_{r(\gamma),\alpha_{n}}^{*})-\langle\nabla\hat{L}_{n}(\beta_{r(\gamma),\alpha_{n}}^{*}),\hat{\Delta}\rangle\\
\geq & c_{8}\|\hat{\Delta}\|_{2}^{2}-c_{9}\alpha_{n}^{-1}\|\hat{\Delta}\|_{2}-c_{10}\alpha_{n}^{-2}-{\displaystyle {\sup_{\hat{\beta}\in \mathcal{H}}}B(\hat{\Delta})-\|\hat{\Delta}\|_{1}\|\nabla\hat{L}_{n}(\beta_{r(\gamma),\alpha_{n}}^{*})\|_{\infty}}\\
\stackrel{p}{\geq} & c_{8}\|\hat{\Delta}\|_{2}^{2}-c_{9}\alpha_{n}^{-1}\|\hat{\Delta}\|_{2}-c_{10}\alpha_{n}^{-2}-2(C-1)\alpha_{n}\frac{\sqrt{s_{n}^{2}log(d)}}{\sqrt{n}}\|\hat{\Delta}\|_{2}-2\sqrt{2}M_2 M_4\|\hat{\Delta}\|_{1}\alpha_{n}\sqrt{\frac{log(d)}{n}}\\
\geq & c_{8}\|\hat{\Delta}\|_{2}^{2}-c_{9}\alpha_{n}^{-1}\|\hat{\Delta}\|_{2}-c_{10}\alpha_{n}^{-2}-2(C-1)\alpha_{n}\frac{\sqrt{s_{n}^{2}log(d)}}{\sqrt{n}}\|\hat{\Delta}\|_{2}-8\sqrt{2}M_2 M_4\|\hat{\Delta}_{S_{n}}\|_{1}\alpha_{n}\sqrt{\frac{log(d)}{n}}\\
\geq & c_{8}\|\hat{\Delta}\|_{2}^{2}-\|\hat{\Delta}\|_{2}(c_{9}\alpha_{n}^{-1}+(2C-2+8\sqrt{2}M_2 M_4)\alpha_{n}\sqrt{\frac{log(d)s_{n}^{2}}{n}})-c_{10}\alpha_{n}^{-2}.
\end{aligned}
\end{equation}
Since $\Psi(\bar{S}_{n})=\displaystyle{\sup_{v\in \bar{S}_{n}\setminus 0}}\dfrac{P(v)}{\|v\|_{2}}\leq\sqrt{s_{n}}$, we can check that the assumptions
of Theorem 4.8 in \cite{hanfang} are all satisfied. Let $\lambda_{n}\asymp\alpha_{n}\sqrt{\frac{log(d)}{n}}$.
Theorem 4.8 in \cite{hanfang} implies that $\|\hat{\Delta}\|_{2}^{2}$ $\stackrel{p}{\leq}(2\lambda_{n}\Psi(\bar{S}_{n})$
$+c_{9}\alpha_{n}^{-1}$ $+(2C-2+8\sqrt{2}M_2 M_4)\alpha_{n}\sqrt{\frac{log(d)s_{n}^{2}}{n}})^{2}/c_{8}^{2}$
$+2c_{10}\alpha_{n}^{-2}/c_{8}$ $\stackrel{p}{\lesssim}\alpha_{n}^{2}log(d)s_{n}^{2}/n+\alpha_{n}^{-2}$.
\end{proof}

\subsection{Proofs of Section \ref{ssec:discrete}}

\begin{proof}[Proof of Theorem \ref{thm:discrete_rate}]
Recall the definition of $\mathcal{M}$ in Assumption \ref{asm:discrete}. Define 
$$\varepsilon=\min\limits _{x_{i}^{T}\beta^{*}>x_{j}^{T}\beta^{*}}(x_{i}-x_{j})^{T}\beta^{*}.$$ $$\mathcal{U}_{t}=\{\beta:\beta_1=1,\forall x_{i},x_{j},\text{ if } (x_{i}^{T}-x_{j}^{T})\beta^{*}>0,\text{ then } (x_{i}^{T}-x_{j}^{T})_{S}\beta_{S}+(x_{i}^{T}-x_{j}^{T})_{S\setminus 1}\mathcal{M}^T\beta_{S^c}>t\}.$$
We let $\alpha_{n}\geq1$, then $\beta^{*}\in \mathcal{U}_{\frac{\epsilon}{\alpha_{n}}}$ according to the definition of $\epsilon$, so $\mathcal{U}_{\frac{\epsilon}{\alpha_{n}}}$
is nonempty. We then prove the theorem by three steps. 

First, let's prove: for large
enough $\alpha_{n}$, $\beta_{r,\alpha_{n}}^{*}\in \mathcal{U}_{\frac{\epsilon}{\alpha_{n}}}$. Clearly
\begin{equation}
\begin{aligned}
C_{n}(\beta) & =E[(Q(X_{i}^{T}\beta^{*})-Q(X_{j}^{T}\beta^{*}))(2F((X_{i}^{T}-X_{j}^{T})\beta\alpha_{n})-1)\mathbbm{1}((X_{i}^{T}-X_{j}^{T})\beta^{*}>0)]\\
\leq & E[(Q(X_{i}^{T}\beta^{*})-Q(X_{j}^{T}\beta^{*}))\cdot1\cdot \mathbbm{1}((X_{i}^{T}-X_{j}^{T})\beta^{*}>0)]=C(\beta^{*}),
\end{aligned}
\end{equation}
and there exists a positive constant $c_{11}$ s.t. 
\begin{equation}
\begin{aligned}
\label{equ:diff_c_beta}
C(\beta^{*})-C_{n}(\beta^{*})= & 2E[(Q(X_{j}^{T}\beta^{*})-Q(X_{i}^{T}\beta^{*}))F((X_{i}^{T}-X_{j}^{T})\beta^{*}\alpha_{n})\mathbbm{1}((X_{i}^{T}-X_{j}^{T})\beta^{*}<0)]\\
\leq & 2E[(Q(X_{j}^{T}\beta^{*})-Q(X_{i}^{T}\beta^{*}))\mathbbm{1}(X_{j}^{T}\beta^{*}>X_{i}^{T}\beta^{*})]/(1+e^{\alpha_{n}\varepsilon})\leq c_{11}e^{-\varepsilon\alpha_{n}}.
\end{aligned}
\end{equation}

If $\beta \notin \mathcal{U}_{\frac{\epsilon}{\alpha_{n}}}$,
then $\exists x_{i}^{T},x_{j}^{T}$, $s.t.$ $(x_{i}^{T}-x_{j}^{T})\beta^{*}>0$
but $(x_{i}^{T}-x_{j}^{T})_{S}\beta_{S}+(x_{i}^{T}-x_{j}^{T})_{S\setminus 1}\mathcal{M}^T\beta_{S^c}<\frac{\varepsilon}{\alpha_{n}}$. Define $\delta =\min\limits _{x_{i}^{T}\beta^{*}>x_{j}^{T}\beta^{*}}(1-F(\varepsilon))(Q(x_{i}^{T}\beta^{*})-Q(x_{j}^{T}\beta^{*}))p(x_{i,S})p(x_{j,S})$, where $p(X_{S})$ is the p.m.f. of $X_{S}$, then $\delta>0$ independent of $d$. Using (B1) in Assumption \ref{asm:discrete}, We have
\begin{equation}
\begin{aligned}
\label{equ:diff_c_beta2}
&C(\beta^{*})-C_{n}(\beta)\\  =&E[(Q(X_{i}^{T}\beta^{*})-Q(X_{j}^{T}\beta^{*}))(2-2F((X_{i}^{T}-X_{j}^{T})\beta\alpha_{n}))\mathbbm{1}((X_{i}^{T}-X_{j}^{T})\beta^{*}>0)]\\
= & E[(Q(X_{i}^{T}\beta^{*})-Q(X_{j}^{T}\beta^{*}))(2-2F((X_{i}-X_{j})_{S}^{T}\beta_{S}\alpha_{n}+(X_{i}-X_{j})_{S\setminus 1}^{T}\mathcal{M}^T\beta_{S^c}\alpha_{n}\\
&+[(X_{i}-X_{j})_{S^{c}}-\mathcal{M}(X_{i}-X_{j})_{S\setminus 1}]^{T}\beta_{S^{c}}\alpha_{n})) \mathbbm{1}((X_{i}^{T}-X_{j}^{T})\beta^{*}>0)]\\
\geq & E[(Q(X_{i}^{T}\beta^{*})-Q(X_{j}^{T}\beta^{*}))(1-F((X_{i}-X_{j})_{S}^{T}\beta_{S}\alpha_{n}+(X_{i}-X_{j})_{S\setminus 1}^{T}\mathcal{M}^T\beta_{S^c}\alpha_{n}))\\
& \mathbbm{1}((X_{i}^{T}-X_{j}^{T})\beta^{*}>0)]\geq p(x_{i,S})p(x_{j,S})(Q(x_{i}^{T}\beta^{*})-Q(x_{j}^{T}\beta^{*}))(1-F(\varepsilon))\geq\delta.
\end{aligned}
\end{equation}
Let $\alpha_{n}>\frac{log(c_{11})-log(\delta)}{\varepsilon}$, then
$C(\beta^{*})-C_{n}(\beta^{*})\leq c_{11}e^{-\varepsilon\alpha_{n}}<\delta\leq C(\beta^{*})-C_{n}(\beta)$,
which implies that $\beta \notin \mathcal{U}_{\frac{\epsilon}{\alpha_{n}}}$
is not a maximizer of $C_{n}(\cdot)$. So $\beta_{r,\alpha_{n}}^{*}\text{\ensuremath{\in}}\mathcal{U}_{\frac{\epsilon}{\alpha_{n}}}$.

Secondly, let's prove the entries of $\beta_{r,\alpha_{n},S^{c}}^{*}$
are all 0. By (B1) in Assumption \ref{asm:discrete},
\begin{equation}
\begin{aligned}
&C_{n}(\beta)\\
=&E[(Q(X_{i}^{T}\beta^{*})-Q(X_{j}^{T}\beta^{*}))(2F((X_{i}-X_{j})^{T}\beta\alpha_{n})-1)\mathbbm{1}((X_{i}-X_{j})^{T}\beta^{*}>0)]\\
= & E[(Q(X_{i}^{T}\beta^{*})-Q(X_{j}^{T}\beta^{*}))(2F((X_{i}-X_{j})_{S}^{T}\beta_{S}\alpha_{n}+(X_{i}-X_{j})_{S\setminus 1}^{T}\mathcal{M}^T\beta_{S^c}\alpha_{n}\\
&+[(X_{i}-X_{j})_{S^{c}}-\mathcal{M}(X_{i}-X_{j})_{S\setminus 1}]^{T}\beta_{S^{c}}\alpha_{n})+2F((X_{i}-X_{j})_{S}^{T}\beta_{S}\alpha_{n}+(X_{i}-X_{j})_{S\setminus 1}^{T}\mathcal{M}^T\beta_{S^c}\alpha_{n}\\
&-[(X_{i}-X_{j})_{S^{c}}-\mathcal{M}(X_{i}-X_{j})_{S\setminus 1}]^{T}\beta_{S^{c}}\alpha_{n})-2)\cdot \mathbbm{1}((X_{i}-X_{j})^{T}\beta^{*}>0)\cdot\\
 & \mathbbm{1}([(X_{i}-X_{j})_{S^{c}}-\mathcal{M}(X_{i}-X_{j})_{S\setminus 1}]^{T}\beta_{S^{c}}>0)].
\end{aligned}
\end{equation}
Let $\beta=\beta_{r,\alpha_{n}}^{*}$ in the formula. And notice that $\forall a>0,t>0,\frac{\partial}{\partial t}(F(a+t)+F(a-t))=\frac{e^{-a+t}}{(e^{t}+e^{-a})^{2}}-\frac{e^{-a+t}}{(1+e^{-a+t})^{2}}$. Because $e^{t}+e^{-a}-(1+e^{-a+t})=(e^{t}-1)(1-e^{-a})>0$, we know
$(F(a+t)+F(a-t))$ attains its maximum at $t=0$ when $a>0$. Therefore, we have

\begin{equation}
\label{equ:symmetric}
\begin{aligned}
&C_{n}(\beta_{r,\alpha_{n}}^{*})\\ 
\leq & 2E[(Q(X_{i}^{T}\beta^{*})-Q(X_{j}^{T}\beta^{*}))(2F((X_{i}-X_{j})_{S}^{T}\beta_{r,\alpha_{n},S}^{*}\alpha_{n}+(X_{i}-X_{j})_{S\setminus 1}^{T}\mathcal{M}^T\beta_{r,\alpha_{n},S^c}^*\alpha_{n})-1)\cdot\\
 & \mathbbm{1}((X_{i}-X_{j})^{T}\beta^{*}>0)\mathbbm{1}([(X_{i}-X_{j})_{S^{c}}-\mathcal{M}(X_{i}-X_{j})_{S\setminus 1}]^{T}\beta_{r,\alpha_{n},S^{c}}^*>0)]\\
= & E[(Q(X_{i}^{T}\beta^{*})-Q(X_{j}^{T}\beta^{*}))(2F((X_{i}-X_{j})_{S}^{T}\beta_{r,\alpha_{n},S}^{*}\alpha_{n}+(X_{i}-X_{j})_{S\setminus 1}^{T}\mathcal{M}^T\beta_{r,\alpha_{n},S^c}^*\alpha_{n})-1)\cdot\\
 & \mathbbm{1}((X_{i}-X_{j})^{T}\beta^{*}>0)],
\end{aligned}
\end{equation}
which means the $\beta$ such that $\beta_1=1$, $\beta_{S\setminus 1}=\beta_{r,\alpha_{n},S\setminus 1}^{*}+\mathcal{M}^T\beta_{r,\alpha_{n},S^c}^*$ and $\beta_{S^c}=\Vec{0}$ will let $C_n(\beta)$ no smaller than $C_n(\beta_{r,\alpha_{n}}^{*})$. So without loss of generality we can set $\beta_{r,\alpha_{n},S^{c}}^{*}=\Vec{0}$ since $\beta_{r,\alpha_{n}}^{*}$ maximizes $C_n(\cdot)$.

Finally, define $F(\hat{\Delta})=\hat{L}_{n}(\beta_{r,\alpha_{n}}^{*}+\hat{\Delta})+\lambda_{n}\|\beta_{r,\alpha_{n}}^{*}+\hat{\Delta}\|_{1}-\hat{L}_{n}(\beta_{r,\alpha_{n}}^{*})-\lambda_{n}\|\beta_{r,\alpha_{n}}^{*}\|_{1}$, then

\begin{equation}
\begin{aligned}
F(\hat{\Delta})\geq & C_{n}(\beta_{r,\alpha_{n}}^{*})-C_{n}(\hat{\beta}_{r,\alpha_{n}})-{\displaystyle {\sup_{\hat{\beta}\in \mathcal{H}}}B(\hat{\Delta})+}\lambda_{n}(\|\beta_{r,\alpha_{n}}^{*}+\hat{\Delta}\|_{1}-\|\beta_{r,\alpha_{n}}^{*}\|_{1})\\
= & C(\beta^{*})-C_{n}(\hat{\beta}_{r,\alpha_{n}})-(C(\beta^{*})-C_{n}(\beta_{r,\alpha_{n}}^{*}))-{\displaystyle {\sup_{\hat{\beta}\in \mathcal{H}}}B(\hat{\Delta})+}\lambda_{n}(\|\beta_{r,\alpha_{n}}^{*}+\hat{\Delta}\|_{1}-\|\beta_{r,\alpha_{n}}^{*}\|_{1}).
\end{aligned}
\end{equation}
If $\hat{\beta}_{r,\alpha_{n}}\notin \mathcal{U}_{\frac{\epsilon}{\alpha_{n}}}$, using (\ref{equ:diff_c_beta}), \ref{equ:diff_c_beta2} and the fact that $\beta_{r,\alpha_{n},S^{c}}^{*}=\Vec{0}$,
\begin{equation}
\begin{aligned}F(\hat{\Delta})\geq &  \delta-c_{11}e^{-\varepsilon\alpha_{n}}-{\displaystyle {\sup_{\hat{\beta}\in \mathcal{H}}}B(\hat{\Delta})+}\lambda_{n}(\|\beta_{r,\alpha_{n}}^{*}+\hat{\Delta}\|_{1}-\|\beta_{r,\alpha_{n}}^{*}\|_{1})\\
\geq & \delta-c_{11}e^{-\varepsilon\alpha_{n}}-{\displaystyle {\sup_{\hat{\beta}\in \mathcal{H}}}B(\hat{\Delta})+}\lambda_{n}(\|\beta_{r,\alpha_{n},S}^{*}+\hat{\Delta}_{S}\|_{1}+\|\hat{\Delta}_{S^{c}}\|_{1}-\|\beta_{r,\alpha_{n},S}^{*}\|_{1}),
\end{aligned}
\end{equation}
then by Lemma \ref{lem:stochastic_error_rate},
\begin{equation}
\begin{aligned}
F(\hat{\Delta})\stackrel{p}{\geq} & \delta-c_{11}e^{-\varepsilon\alpha_{n}}+\lambda_{n}(-\|\hat{\Delta}_{S}\|_{1}+\|\hat{\Delta}_{S^{c}}\|_{1})-c_{12}\|\hat{\Delta}\|_{2}\alpha_{n}\sqrt{log(d)s^{2}/n}\\
\geq & -\lambda_{n}\|\hat{\Delta}_{S}\|_{1}+\delta-c_{11}e^{-\varepsilon\alpha_{n}}-c_{12}\|\hat{\Delta}\|_{2}\alpha_{n}\sqrt{log(d)s^{2}/n}\\
\geq & -\lambda_{n}\|\hat{\Delta}\|_{2}\sqrt{s}+\delta-c_{11}e^{-\varepsilon\alpha_{n}}-c_{12}\|\hat{\Delta}\|_{2}\alpha_{n}\sqrt{log(d)s^{2}/n}.
\end{aligned}
\end{equation}
Choose $\lambda_{n}=\frac{\delta}{3r\sqrt{s}}$ and $\alpha_{n}=max(1,\frac{log(c_{11})-log(\delta)+log(3)}{\varepsilon})$,
then this satisfies $\alpha_{n}>\frac{log(c_{11})-log(\delta)}{\varepsilon}$
and $\lambda_{n}\gtrsim\alpha_{n}\sqrt{\frac{log(d)}{n}}$. So 
\begin{equation}
F(\hat{\Delta})\stackrel{p}{\geq}\delta/3-c_{12}\|\hat{\Delta}\|_{2}\alpha_{n}\sqrt{log(d)s^{2}/n}\geq\delta/3-max(1,\frac{log(3c_{11}/\delta)}{\varepsilon})c_{12}r\sqrt{log(d)s^{2}/n}
\end{equation}
If $log(d)s^{2}/n\longrightarrow0$, then $F(\hat{\Delta})\stackrel{p}{>}0$,
which contradicts with the fact that $\hat{\beta}_{r,\alpha_{n}}$
minimizes $\hat{L}_{n}(\beta)+\lambda_{n}P(\beta)$ in $\mathcal{H}$.
Hence $P(\hat{\beta}_{r,\alpha_{n}}\in \mathcal{U}_{\frac{\epsilon}{\alpha_{n}}})\longrightarrow1$.

Now let's consider $\hat{\beta}_{r,\alpha_{n},S^{c}}$. Truncate the $S^c$ part of $\hat{\beta}_{r,\alpha_{n}}$ to be 0 and call this new vector as $\widetilde{\hat{\beta}}_{r,\alpha_{n}}$.
By definition, $\hat{L}_{n}(\hat{\beta}_{r,\alpha_{n}})+\lambda_{n}\|\hat{\beta}_{r,\alpha_{n}}\|_{1}$$\leq\hat{L}_{n}(\widetilde{\hat{\beta}}_{r,\alpha_{n}})+\lambda_{n}\|\widetilde{\hat{\beta}}_{r,\alpha_{n}}\|_{1}$,
so $\lambda_{n}\|\hat{\beta}_{r,\alpha_{n},S^{c}}\|_{1}$ $\leq\hat{L}_{n}(\widetilde{\hat{\beta}}_{r,\alpha_{n}})$
$-\hat{L}_{n}(\hat{\beta}_{r,\alpha_{n}})$. Similar to (\ref{equ:symmetric}), we know $L_{n}(\widetilde{\hat{\beta}}_{r,\alpha_{n}})-L_{n}(\hat{\beta}_{r,\alpha_{n}})\leq0$,
so
\begin{equation}
\lambda_{n}\|\hat{\beta}_{r,\alpha_{n},S^{c}}\|_{1}\leq B(\hat{\beta}_{r,\alpha_{n}}-\widetilde{\hat{\beta}}_{r,\alpha_{n}})\stackrel{p}{\lesssim}(\frac{s^{2}\log(d)}{n})^{1/2},
\end{equation}
and therefore $\|\hat{\beta}_{r,\alpha_{n},S^{c}}\|_{1}\stackrel{p}{\lesssim}(\frac{s^{3}\log(d)}{n})^{1/2}$.
\end{proof}

\bibliographystyle{apalike}

\bibliography{mybib}

\end{document}